\documentclass[authoryear,12pt]{article}
\usepackage[english]{babel}
\usepackage{amsmath}
\usepackage{latexsym}
\usepackage{amssymb}
\usepackage{graphicx}
\usepackage{hyperref}
\usepackage{natbib}
\usepackage{epsfig}
\usepackage[normalem]{ulem}
\usepackage{setspace}

\newcommand{\T}{\top}
\newcommand{\pr}{\mathrm{pr}}

\newcommand{\Mob}{M\"{o}bius}

\newcommand{\ci}{\mbox{$\perp \! \! \! \perp$}}
\newtheorem{thm}{Theorem}[section]
\newtheorem{lem}[thm]{Lemma}
\newtheorem{prop}[thm]{Proposition}
\newtheorem{cor}[thm]{Corollary}

\newtheorem{defn}{Definition}[section]

\newtheorem{exmp}{Example}[section]

\newenvironment{proof}{\begin{trivlist}\item[] \mbox{\textit{Proof.}}}
{\hfill$\Box$ \end{trivlist}}

\usepackage{color} 
\begin{document}
\title{Log-mean linear regression models for binary responses with an application to multimorbidity}
\author{Monia Lupparelli\footnote{Department of Statistical Sciences, University of Bologna, IT
(\texttt{monia.lupparelli@unibo.it})} \hspace{0.1cm} and \hspace{0.01cm}
Alberto Roverato\footnote{Department of Statistical Sciences, University of Bologna, IT
(\texttt{alberto.roverato@unibo.it})}
}

\date{March 30, 2016}
\maketitle

\begin{abstract}
In regression models for categorical data a linear model is typically related to the response variables via a transformation of probabilities called the link function.  We introduce an approach based on two link functions for binary data named \emph{log-mean} (LM) and \emph{log-mean linear} (LML), respectively. The choice of the link function plays a key role for the interpretation of the model, and our approach is especially appealing in terms of interpretation of the effects of covariates on the association of responses.
Similarly to Poisson regression, the LM and LML regression coefficients of single outcomes are log-relative risks,
and we show that the relative risk interpretation is maintained also in the regressions of the association of responses. Furthermore,
certain collections of zero LML regression coefficients imply that the relative risks for joint responses factorize with respect to the
corresponding relative risks for marginal responses.  This work is motivated by the analysis of a dataset obtained from a case-control study
aimed to investigate the effect of HIV-infection on \emph{multimorbidity}, that is simultaneous presence of two or more noninfectious commorbidities in one patient.
\bigskip

\noindent\emph{Keywords}:
Categorical data; Link function; Multimorbidity pattern; Relative risk; Response association.

\end{abstract}

\section{Introduction}\label{sec.intro}
In many research fields it is required to model the dependence of a collection of response variables on one or more explanatory variables. \citet[][Sec.~6.5]{mcc-nel:1989} specified that in this context there are  typically three lines of inquiry: (i) the dependence structure of each response marginally on covariates, (ii) a model for the joint distribution of all responses and (iii) the joint dependence of response variables on covariates. When in a regression model responses are categorical,  a linear model is typically related to the response variables via a transformation of probabilities called the link function, and the choice of the link function plays a key role for the interpretation of the model along the lines (i) to (iii). We refer to \citet{tutz2011regression} and \citet{agresti2013categorical} for a full account of regression models for categorical data; see also \citet[Section~5]{Ekh-al:2000} for a review of some link functions commonly used in the binary case.

This work is motivated by a research aimed to investigate the effect of HIV-infection on \emph{multimorbiditiy} which is defined as the co-occurrence of two or more chronic medical conditions in one person. It is well known that multimorbidity is associated with age and, furthermore, that HIV-infected patients experience an increased prevalence of noninfectious comorbidities, compared with the general population.  \citet{Guaraldi-al:2011} considered a dataset obtained from a cross-sectional retrospective case-control study, and investigated the effect of HIV-infection on the prevalence of a set of noninfectious chronic medical conditions by applying an univariate regression to each response. However, multimorbidity is characterised by complex interactions of co-existing diseases and to gain relevant insight it is necessary to use a multivariate approach aimed to investigate the effect of HIV on the way different chronic conditions associate. The main scientific objective of this study is thus the line of enquiry (iii). However, to the best of our knowledge, this line has never been explicitly addressed in the literature, and this paper is fully devoted to this issue.

The application we consider naturally requires a \emph{marginal modelling} approach because the main interest is for the effect of HIV on the marginal association of subsets of comorbidities; see \citet[][Chapter~13]{tutz2011regression} and \citet[][Chapter~12]{agresti2013categorical}.
For this reason, we focus on the case where the link function satisfies \emph{upward compatibility}, that is every association term among responses can be computed in the relative marginal distribution.
In this way, the parameterization of the response variables will include terms that can be regarded as \emph{single outcomes}, computed marginally on univariate responses, and terms that can be regarded as \emph{association outcomes}, hereafter referred to as \emph{response associations}, which are computed marginally on subsets of responses.
Regression models typically include coefficients encoding the effect of the covariates, as well as of interactions of covariates, on response associations and difficulties involve both the interpretation of the response associations and the interpretation of the relevant regression coefficients. More seriously, the effect of a covariate on a response association might be \emph{removable} in the sense that it disappears when a different link function is used. It is therefore crucially relevant to be able to define models with interpretable regressions coefficients; see \citet{berrington2007interpretation} for a review on statistical interactions, with emphasis on interpretation.

In marginal modelling a central role is played by the \emph{multivariate logistic} regression model because it maintains a marginal logistic regression interpretation for the single outcomes \citep{mcc-nel:1989,glo-mcc:1995}. Nevertheless,  this family of regression models does not provide a satisfying answer to the  multimorbidity application. This is due to the fact that, although the regression coefficients for the single outcomes can be interpreted in terms of odds ratios, this feature does not translate to the higher order regressions where both the response association and the relevant regression coefficients are high-level log-linear parameters which are difficult to interpret.

We consider two different parameterizations, the \emph{log-mean} (LM) \citep{drton2008binary,drton2009discrete} and the \emph{log-mean linear} (LML) parameterization \citep{rov-lup-lar:2013,roverato2015log}, and investigate the use of these parameterizations as link functions. In this way, we introduce an approach where, similarly to Poisson regression, regression coefficients can be interpreted in terms of relative risks. Furthermore, and more interestingly, the relative risk interpretation can be extended from the regressions of the single outcomes to the regressions of the response associations, thereby providing interpretable coefficients. The LM and the LML links can be used to specify the same classes of submodels but the LML link has the advantage that relevant submodels can be specified by setting regression coefficients to zero.  Specifically, we show that certain collections of zero LML regression coefficients imply that the relative risks for joint responses factorize with respect to the corresponding relative risks for marginal responses.

The paper is organized as follows. In Section~\ref{SEC:multimorbidity} we describe the motivating problem concerning the analysis of the multimorbidity data.
Section~\ref{sec:back} gives the background concerning the theory of regression for multivariate binary responses, as required for this paper. Section~\ref{sec.LM-LML} introduces the  LM and the LML regression models and describes the relevant properties of these models.  The analysis of multimorbidity data is carried out in Section~\ref{sec:app} and, finally,  Section~\ref{sec:discuss} contains a  discussion.

\section{Motivating problem: multimorbidity in HIV-positive patients}\label{SEC:multimorbidity}
Antiretroviral therapy (ART) for human immunodeficiency virus (HIV) infection has been a great medical success story. Nowadays, in countries with good access to treatment, clinical AIDS is no longer the inevitable outcome of HIV infection and this disease, previously associated with extremely high mortality rates, is now generally thought of
as a chronic condition \citep{Mocroft200322,may2011life}. Despite a marked increase in life expectancy, mortality rates among HIV-infected persons remain higher than those seen in the general population. Some of the excess mortality observed among HIV-infected persons can be directly attributed to illnesses that occur as a consequence of immunodeficiency, however, more than half of the deaths observed in recent years among ART-experienced HIV infected patients are attributable to noninfectious comorbidities (NICMs) \citep{phillips2008role,Guaraldi-al:2011}.

\emph{Multimorbidity} is defined as the co-occurrence of two or more chronic medical conditions in one person that is, for HIV positive patients, as the simultaneous presence of two or more NICMs.  Multimorbidity, which is associated with age, is perhaps the most common \lq\lq{}disease pattern\rq\rq{} found among the elderly and, for this reason, it is turning into a major medical issue for both individuals and health care providers \citep{marengoni2011aging}. It is well known that HIV-infected patients experience an increased prevalence of NICMs, compared with the general population, and it has been hypothesized that such increased prevalence is the result of premature aging of HIV-infected patients \citep{deeks2009hiv,shiels2010age,Guaraldi-al:2011}.
Multimorbidity is characterised by the co-occurrence of NICMs and therefore investigating the effect of HIV-infection on multimorbidity requires to investigate the effect of HIV on the way different chronic conditions associate.

The dataset we analyse here comes from a study of \citet{Guaraldi-al:2011} who investigated the effect of HIV-infection on the prevalence of a set of  noninfectious chronic medical conditions. Data were obtained from a cross-sectional retrospective case-control study with sample size $n=11\,416$ (2854 cases and 8562 controls). Cases were ART-experienced HIV-infected patients older than 18 years of age who were consecutively enrolled at the Metabolic Clinic of Modena University in Italy from 2002 to 2009. Control subjects were matched according to age, sex, race (all white), and geographical area. The observed variables include both a set of binary (response) variables encoding the presence of NICMs of interest and a set of context and clinical covariates.  See \citet{Guaraldi-al:2011} for details and additional references.
\section{Background and notation}\label{sec:back}
\subsection{\Mob\ inversion}\label{subsec:mob}
In this subsection we introduce the notation used for matrices and recall a well known result named \emph{\Mob\ inversion} that will be extensively used in the following.

For two finite sets $V$ and $U$, with $|V|=p$ and $|U|=q$, we write $\theta=\{\theta_{D}(E)\}_{D\subseteq V, E\subseteq U}$ to denote a $2^{p}\times 2^{q}$ real matrix with rows and columns indexed by the subsets of $V$ and $U$, respectively. Furthermore, we will write $\theta(E)$ to denote the column of $\theta$ indexed by $E\subseteq U$ and $\theta_{D}$ to denote the row of $\theta$ indexed by $D\subseteq V$. Note that the notation we use may be easier to read if one associates $D$ with $D$iseases and $E$ with $E$xposure.
\begin{exmp}[Matrix notation.]\label{EXA:matrix.notation}
For the case $V=\{b, c, d\}$ and $U=\{h, a\}$ the matrix $\theta$ has eight rows indexed by the subsets $\emptyset$, $\{b\}$, $\{c\}$, $\{d\}$ $\{b, c\}$, $\{b, d\}$, $\{c, d\}$, $\{b, c, d\}$ and four columns indexed by the subsets $\emptyset$, $\{h\}$, $\{a\}$, $\{h, a\}$. The matrix $\theta$, with row and column indexes, is given below; note that we use the suppressed notation $\theta_{bc}(ha)$ to denote $\theta_{\{b,c\}}(\{h, a\})$, and similarly for the other quantities.
\begin{eqnarray*}
\begin{array}{ll}
\begin{array}{llll}
\hspace{1.3cm}\emptyset & \hspace{.7cm}\{h\} & \hspace{.7cm}\{a\} & \hspace{.4cm}\{h,a\}
\end{array}
& ~ \\[2mm]
\theta=
\left(
\begin{array}{cccc}
  \theta_{\emptyset}(\emptyset) & \theta_{\emptyset}(h) & \theta_{\emptyset}(a) & \theta_{\emptyset}(ha) \\
  \theta_{b}(\emptyset) & \theta_{b}(h) & \theta_{b}(a) & \theta_{b}(ha) \\
  \theta_{c}(\emptyset) & \theta_{c}(h) & \theta_{c}(a) & \theta_{c}(ha) \\
  \theta_{d}(\emptyset) & \theta_{d}(h) & \theta_{d}(a) & \theta_{d}(ha) \\
  \theta_{bc}(\emptyset) & \theta_{bc}(h) & \theta_{bc}(a) & \theta_{bc}(ha)\\
  \theta_{bd}(\emptyset) & \theta_{bd}(h) & \theta_{bd}(a) & \theta_{bd}(ha)\\
  \theta_{cd}(\emptyset) & \theta_{cd}(h) & \theta_{cd}(a) & \theta_{cd}(ha)\\
  \theta_{bcd}(\emptyset) & \theta_{bcd}(h) & \theta_{bcd}(a) & \theta_{bcd}(ha)
\end{array}
\right)
&\hspace*{-3mm}
\begin{array}{l}
\emptyset\\
\{b\} \\
\{c\} \\
\{d\} \\
\{b,c\} \\
\{b,d\} \\
\{c,d\} \\
\{b,c,d\}
\end{array}
\end{array}%
\end{eqnarray*}
Furthermore, the rows and columns of $\theta$ are denoted as follows.
\begin{eqnarray*}
\theta =
\left(
\begin{array}{l}
\theta_{\emptyset}\\
\theta_{b}\\
\theta_{c}\\
\theta_{d}\\
\theta_{bc}\\
\theta_{bd}\\
\theta_{cd}\\
\theta_{bcd}\\
\end{array}
\right)
=
\left(
\begin{array}{llll}
\theta(\emptyset) & \theta(h) & \theta(a) & \theta(ha)
\end{array}
\right)
\end{eqnarray*}
\end{exmp}

This matrix notation is not standard in the literature concerning categorical data but, in the case of regression models for binary data, it allows us to provide a compact representation of model parameters in matrix form and to compute alternative parameterizations, as well as regression coefficients, by direct application of \Mob\ inversion.

Let $\omega$ be another real matrix indexed by the subsets of  $V$ and $U$. For a subset $D\subseteq V$ \Mob\ inversion states that
\begin{equation}\label{EQN.mon.Mobius}
\theta_{D}(E)=\sum_{E^{\prime} \subseteq E} \omega_{D}(E^{\prime}),\;\; \forall E \subseteq U
\quad\Leftrightarrow\quad \omega_{D}(E)=\sum_{E^{\prime} \subseteq E} (-1)^{|E\backslash E^{\prime}|} \theta_{D}(E^{\prime}),\;\; \forall E \subseteq U;
\end{equation}
see, among others, \citet[][Appendix~A]{lau:1996}. Let $Z_{U}$ and $M_{U}$ be two $(2^{q} \times 2^q)$ matrices
with entries indexed by the subsets of $U \times U$ such that the entry of $Z$ indexed by the pair $E,H\subseteq U$ is equal to $1(E \subseteq H)$ and the corresponding entry of $M$ is equal to $(-1)^{|H \backslash E|} 1(E \subseteq H)$, where $1(\cdot)$ denotes the indicator function. Then, the equivalence~(\ref{EQN.mon.Mobius}) can be written in matrix form as
$\theta_{D}= \omega_{D}Z_{U}\; \Leftrightarrow\; \omega_{D} =\theta_{D}M_{U}$
and \Mob\ inversion follows by noticing that $M_{U}=Z^{-1}_{U}$. Note that it is straightforward to extend this result to the matrices $\omega$ and $\theta$ as
\begin{equation}\label{EQN.mon.Mobius.matrix}
\theta= \omega Z_{U} \quad \Leftrightarrow\quad \omega =\theta M_{U}
\end{equation}
and, furthermore, that it makes sense to consider \Mob\ inversion also with respect to the columns of $\omega$ and $\theta$ so that $\theta= Z_{V}^{\T}\omega \; \Leftrightarrow\; \omega = M^{\T}_{V}\theta$.
\subsection{Multivariate binary response models}\label{sec.mult-bin-mod}
Let $Y_{V}=(Y_{v})_{v\in V}$ be a binary random vector of response variables with entries indexed by $V$ and $X_{U}=(X_{u})_{u\in U}$ a vector of binary covariates with entries indexed by $U$. Without loss of generality, we assume that $Y_{V}$ and $X_{U}$ take value in $\{0,1\}^{p}$ and $\{0,1\}^{q}$, respectively.  The values of covariates denote different observational or experimental conditions and we assume that, for every $x_{U}\in \{0,1\}^{q}$, the distribution of $Y_{V}|(X_{U}=x_{U})$ is multivariate Bernoulli. Furthermore, we assume that the latter distributions are independent across conditions and that, when $X_{U}$ is regarded as a random vector, then also $(Y_{V}, X_{U})$ follows a multivariate Bernoulli distribution. We can write the probability distributions of $Y_{V}|X_{U}$ by means of a matrix $\pi=\{\pi_{D}(E)\}_{D\subseteq V, E\subseteq U}$ where, for every $E\subseteq U$, the column vector $\pi(E)$ is the probability distribution of $Y_{V}$ given $(X_{E}=1_{E},X_{U\backslash E}=0_{U\backslash E})$ and, more specifically, $\pi_{D}(E)=\pr(Y_{D}=1_{D},Y_{V\backslash D}=0_{V\backslash D}\mid X_{E}=1_{E},X_{U\backslash E}=0_{U\backslash E})$. We assume that  all the entries of $\pi$ are strictly positive. In the following, for a subset $D\subseteq V$ we use the suppressed notation $Y_{D}=1$ to denote $Y_{D}=1_{D}$ and similarly for $Y_{D}=0$ and the subvectors of $X_{U}$. Given three random vectors $X$, $Y$ and $Z$, we write $X\ci Y|Z$ to say that $X$ is independent of $Y$ given $Z$ \citep{dawid1979conditional} or, in the case where $Y$ and $Z$ are not random, that the conditional distribution of $X$ given $Y$ and $Z$ does not depend on $Y$.

In regression models for categorical responses a linear regression is typically related to the response variables via a link function $\theta(\pi)$. All the link functions considered in this paper are such that, for every $E\subseteq U$, the vector $\theta(E)=\theta(\pi(E))$ parameterizes the distribution of $Y_{V}|(X_{E}=1,X_{U\backslash E}=0)$.
In this way, the link function induces a matrix $\theta=\{\theta_{D}(E)\}_{D\subseteq V, E\subseteq U}$ and the associated linear regression, in the saturated case, has form
\begin{eqnarray}\label{EQN:betadef}
\theta_{D}(E)=\sum_{E^{\prime}\subseteq E} \beta^{\langle\theta\rangle}_{D}(E^{\prime})
\qquad\mbox{for every }D\subseteq V, E\subseteq U
\end{eqnarray}
where $\beta^{\langle\theta\rangle}=\{\beta^{\langle\theta\rangle}_{D}(E)\}_{D\subseteq V, E\subseteq U}$ is a matrix of regression coefficients. It follows from
(\ref{EQN.mon.Mobius}) and (\ref{EQN.mon.Mobius.matrix}) that (\ref{EQN:betadef}) can be written in matrix form as
\begin{eqnarray}\label{EQN:beta-theta}
\theta=\beta^{\langle\theta\rangle}Z_{U}\qquad\mbox{so that}\qquad \beta^{\langle\theta\rangle}=\theta M_{U}.
\end{eqnarray}
The regression setting (\ref{EQN:betadef}) involves multivariate combinations of both the responses and the covariates and for this reason it is important to explicitly distinguish between the $D$-\emph{response associations} which are given by $\theta_{D}$ for $D\subseteq V$ with $|D|>1$, and the $E$-\emph{covariate interactions} given by $\beta^{\langle\theta\rangle}(E)$ for $E\subseteq U$ with $|E|>1$. Hence,  $\beta_D^{\langle\theta\rangle}(E)$ encodes the effect of the $E$-covariate interaction on the $D$-response association, given a fixed level of $X_{U \setminus E}$ that, without loss of generality, in our approach is the zero level.
\begin{exmp}[Matrix notation continued.]
In all the examples of this paper we use the variables of the multimorbidity data as given in Section~\ref{sec:app}. Specifically, we consider three of the response variables, which are $Y_{b}=\;$\emph{Bone fracture}, $Y_{c}=\;$\emph{Cardiovascular disease} and $Y_{d}=\;$\emph{Diabetes}, with level 1 encoding the presence of the disease, and the most relevant covariates, that is, $X_{h}=\,$\emph{HIV} with the level 1 encoding the presence of the infection and $X_{a}=\,$\emph{Age} with the value 1 for patients aged 45 or more.
Hence, if the link function is the matrix $\theta$ given in Example~\ref{EXA:matrix.notation}, then the matrix of regression coefficients $\beta^{\langle\theta\rangle}$ is
\begin{eqnarray*}
\beta^{\langle\theta\rangle}=
\left(
\begin{array}{llll}
  \beta^{\langle\theta\rangle}_{\emptyset}(\emptyset) & \beta^{\langle\theta\rangle}_{\emptyset}(h) & \beta^{\langle\theta\rangle}_{\emptyset}(a) & \beta^{\langle\theta\rangle}_{\emptyset}(ha) \\[1.5mm]
  \beta^{\langle\theta\rangle}_{b}(\emptyset) & \beta^{\langle\theta\rangle}_{b}(h) & \beta^{\langle\theta\rangle}_{b}(a) & \beta^{\langle\theta\rangle}_{b}(ha) \\[1.5mm]
  \beta^{\langle\theta\rangle}_{c}(\emptyset) & \beta^{\langle\theta\rangle}_{c}(h) & \beta^{\langle\theta\rangle}_{c}(a) & \beta^{\langle\theta\rangle}_{c}(ha) \\[1.5mm]
  \beta^{\langle\theta\rangle}_{d}(\emptyset) & \beta^{\langle\theta\rangle}_{d}(h) & \beta^{\langle\theta\rangle}_{d}(a) & \beta^{\langle\theta\rangle}_{d}(ha) \\[1.5mm]
  \beta^{\langle\theta\rangle}_{bc}(\emptyset) & \beta^{\langle\theta\rangle}_{bc}(h) & \beta^{\langle\theta\rangle}_{bc}(a) & \beta^{\langle\theta\rangle}_{bc}(ha)\\[1.5mm]
  \beta^{\langle\theta\rangle}_{bd}(\emptyset) & \beta^{\langle\theta\rangle}_{bd}(h) & \beta^{\langle\theta\rangle}_{bd}(a) & \beta^{\langle\theta\rangle}_{bd}(ha)\\[1.5mm]
  \beta^{\langle\theta\rangle}_{cd}(\emptyset) & \beta^{\langle\theta\rangle}_{cd}(h) & \beta^{\langle\theta\rangle}_{cd}(a) & \beta^{\langle\theta\rangle}_{cd}(ha)\\[1.5mm]
  \beta^{\langle\theta\rangle}_{bcd}(\emptyset) & \beta^{\langle\theta\rangle}_{bcd}(h) & \beta^{\langle\theta\rangle}_{bcd}(a) & \beta^{\langle\theta\rangle}_{bcd}(ha)
\end{array}
\right)
=
\left(
\begin{array}{llll}
  \beta^{\langle\theta\rangle}_{\emptyset}      \\[1.5mm]
  \beta^{\langle\theta\rangle}_{b}   \\[1.5mm]
  \beta^{\langle\theta\rangle}_{c}   \\[1.5mm]
  \beta^{\langle\theta\rangle}_{d}   \\[1.5mm]
  \beta^{\langle\theta\rangle}_{bc}  \\[1.5mm]
  \beta^{\langle\theta\rangle}_{bd}  \\[1.5mm]
  \beta^{\langle\theta\rangle}_{cd}  \\[1.5mm]
  \beta^{\langle\theta\rangle}_{bcd}
\end{array}
\right)
\end{eqnarray*}
Equation (\ref{EQN:betadef}) states that there exists a one-to-one \Mob\ inversion relationship between every row $\theta_{D}$ of $\theta$ and the corresponding row $\beta^{\langle\theta\rangle}_{D}$ of $\beta^{\langle\theta\rangle}$.
For instance, $\theta_{b}$ is the link function of the marginal distribution of bone fracture, and each of its entries, $\theta_{b}(E)$ for $E\subseteq U$, can be computed by considering $\beta_{b}^{\langle\theta\rangle}$ and then by taking the sum of its entries which are indexed by a subset of $E$. More specifically, for $E=\{h\}$ it holds that $\theta_{b}(h)=\beta^{\langle\theta\rangle}_{b}(\emptyset)+\beta^{\langle\theta\rangle}_{b}(h)$
whereas for $E=\{h, a\}$ it holds that
$\theta_{b}(ha)=\beta^{\langle\theta\rangle}_{b}(\emptyset)+\beta^{\langle\theta\rangle}_{b}(h)+\beta^{\langle\theta\rangle}_{b}(a)+\beta^{\langle\theta\rangle}_{b}(ha)$.
Hence, $\beta^{\langle\theta\rangle}_{b}(h)$ and $\beta^{\langle\theta\rangle}_{b}(ha)$ encode the effect of HIV and of the interaction of HIV and age, respectively, on bone fracture. Similarly,  $\theta_{bd}$ is the response association of bone fracture and diabetes. Hence, for $E=\{h\}$ it holds that $\theta_{bd}(h)=\beta^{\langle\theta\rangle}_{bd}(\emptyset)+\beta^{\langle\theta\rangle}_{bd}(h)$
whereas for $E=\{h, a\}$ it holds that
$\theta_{bd}(ha)=\beta^{\langle\theta\rangle}_{bd}(\emptyset)+\beta^{\langle\theta\rangle}_{bd}(h)+\beta^{\langle\theta\rangle}_{bd}(a)+\beta^{\langle\theta\rangle}_{bd}(ha)$.
Here, $\beta^{\langle\theta\rangle}_{bd}(h)$ and $\beta^{\langle\theta\rangle}_{bd}(ha)$ encode the effect of HIV and of the interaction of HIV and age, respectively, on the $\{b,d\}$-response association $\theta_{bd}$.
\end{exmp}

There exists an extensive literature on models for categorical data analysis but, remarkably,
\citet{lang:1996}, extending previous work by \citet{lang:agre:94}, introduced a very general method to specify regression models for categorical data thereby defining an extremely broad class of models  named \emph{generalized log-linear models}. This includes, as special cases, many of the existing models for multiple categorical responses such as log-linear and, more generally, marginal log-linear models \citep{ber-al:2009}. In particular, in marginal regression modelling a relevant instance within this class is obtained when $\theta(\cdot)$ is the  multivariate logistic link function denoted by $\eta(\cdot)$ \citep{mcc-nel:1989,glo-mcc:1995,mole-lesa:1999,bar-al:07,ber-al:2009,mar-lup:2011}. This induces the parameterization  $\eta=\{\eta_{D}(E)\}_{D\subseteq V, E\subseteq U}$ where $\eta_{D}(E)$ is the $|D|$-way log-linear interaction computed in the margin $Y_{D}|(X_{E}=1,X_{U\backslash E}=0)$ and, more specifically, $\eta_{\{v\}}$ is the usual logistic link for every $v\in V$. \citet{mar-lup:2011} considered the regression framework (\ref{EQN:betadef}) with $\theta=\eta$ and showed that $\beta^{\langle\eta\rangle}_{D}(E)$ is the $|D\cup E|$-way log-linear interaction computed in the distributions of $(Y_{D}, X_{U})$.

\begin{exmp}[Multivariate logistic regression]
Consider the case $V=\{b, c\}$ and $U=\{h\}$.
The multivariate logistic parameters are
\begin{eqnarray*}
\eta_{i}(\emptyset)=\log\frac{\pr(Y_{i}=1|X_{h}=0)}{1-\pr(Y_{i}=1|X_{h}=0)}
\qquad\mbox{and}\qquad
\eta_{i}(h)=\log\frac{\pr(Y_{i}=1|X_{h}=1)}{1-\pr(Y_{i}=1|X_{h}=1)}
\end{eqnarray*}
for $i\in \{b, c\}$, which are logit links. Furthermore, if we denote by $OR(Y_{bc}|\emptyset)$ and by $OR(Y_{bc}|h)$ the odds ratio between bone fracture and cardiovascular disease for HIV-negative and HIV-positive patients, respectively, then
\begin{eqnarray*}
\eta_{bc}(\emptyset)=\log OR(Y_{bc}|\emptyset)
\qquad\mbox{and}\qquad
\eta_{bc}(h)=\log OR(Y_{bc}|h)
\end{eqnarray*}
so that the coefficients for the regression of $\eta_{bc}$ are
\begin{eqnarray*}
\beta^{\langle\eta\rangle}_{bc}(\emptyset)=\log OR(Y_{bc}|\emptyset)
\qquad\mbox{and}\qquad
\beta^{\langle\eta\rangle}_{bc}(h)=\log\frac{OR(Y_{bc}|h)}{OR(Y_{bc}|\emptyset)}.
\end{eqnarray*}
\end{exmp}

We now move from the saturated model to submodels defined by means of linear constraints on the regression coefficients and, more specifically, to submodels with regression coefficients equal to zero. If for every  subset of $U$, denoted as $E$, that has non-empty intersection with $U^\prime\subseteq U$, it holds that $\beta^{\langle\theta\rangle}_{D}(E)=0$, then the covariates $X_{U^{\prime}}$ have no-effect on $\theta_{D}$. The following lemma states a connection between no-effect of a subset of variables and linear constraints in $\theta$.
\begin{lem}\label{THM:lemma.beta.theta}
Let $\theta=\{\theta_{D}(E)\}_{D\subseteq V,E\subseteq U}$ be a real matrix with entries indexed by two nonempty sets $V$ and $U$. If $\beta^{\langle\theta\rangle}=\theta M_{U}$ and $U^{\prime}\subseteq U$ then the following are equivalent, for every $D\subseteq V$, as they both say that $X_{U^{\prime}}$ has no effect on $\theta_{D}$;
\begin{itemize}
\item[(i)] $\beta^{\langle\theta\rangle}_{D}(E)=0$ for every $E\subseteq U$ such that $E\cap U^{\prime}\neq\emptyset$;
\item[(ii)] $\theta_{D}(E_{1})=\theta_{D}(E_{2})$ for every $E_{1},E_{2}\subseteq U$ such that $E_{1}\backslash U^{\prime}=E_{2}\backslash U^{\prime}$.
\end{itemize}
\end{lem}
\begin{proof}
See Appendix A.
\end{proof}
\section{Log-mean and log-mean linear regression models}\label{sec.LM-LML}
The multivariate Bernoulli distribution belongs to the natural exponential family, and the mean parameter associated with the distribution $\pi(E)$ is $\mu(E)$ where $\mu_{D}(E)=\pr(Y_{D}=1|X_{E}=1, X_{U\backslash E}=0)$ for $D\subseteq V$. \citet{drton2008binary} used the mean parameter to parameterize graphical models of marginal independence and called it the  \Mob\ parameter because $\mu(E)=Z_{V}\pi(E)$ for every $E\subseteq V$. Subsequently, \citet{drton2009discrete} used the matrix  $\mu=Z_{V}\pi$  to parameterize regression graph models; see also \citet{rov-lup-lar:2013} and \citet{roverato2015log}.  The mean parameterization has the disadvantage that
submodels of interest are defined by, non-linear, multiplicative constraints. \citet{rov-lup-lar:2013} introduced the log-mean linear parameter $\gamma$ defined as a log-linear expansion of the mean parameters, formally,
\begin{eqnarray}\label{EQN:definition.of.gamma}
  \gamma=M^{\T}_{V}\log(\mu),
\end{eqnarray}
and showed that this approach improves on the mean parameterization as submodels of interest can be specified by setting certain zero log-mean linear interactions. We remark that both the mean and the log-mean linear parameters are not variation independent, in the sense that  setting some parameters to particular values may restrict the valid range of other parameters. As a consequence, unlike variation independent parameterizations, it might be more difficult to interpret separately each parameter.

Here, we show that the analysis of multimorbidity data can be effectively approached by  applying the theory described in the previous section to the log-mean (LM) and the log-mean linear (LML) parameterizations to develop the LM and the LML regression model, respectively.

\subsection{Log-mean regression}
The LM regression model is obtained by setting $\theta$ in (\ref{EQN:betadef}) equal to the logarithm of the mean parameter, $\log(\mu)$, so that, in the saturated case, $\log(\mu)=\beta^{\langle\mu\rangle}Z_{U}$ that,
in its extended form, is
\begin{eqnarray}\label{EQN:mob.regression}
  \log\mu_{D}(E)=\sum_{E^{\prime}\subseteq E}\beta^{\langle\mu\rangle}_{D}(E^{\prime})\quad\mbox{ for every }D\subseteq V, E\subseteq U
\end{eqnarray}
where  for $D=\emptyset$ the equation (\ref{EQN:mob.regression}) is trivial for every $E\subseteq U$ because $\mu_{\emptyset}$ is a row vector of ones so that both $\log(\mu_{\emptyset})=0$ and $\beta^{\langle\mu \rangle}_{\emptyset}=\log(\mu_{\emptyset})M_{U}=0$.

LM regression can be regarded as one of the possible alternative ways to parameterize the distribution of $Y_{V}|X_{U}$, but it is of special interest for the application considered in this paper. This can be seen by noticing that one can write $\mu_{D}(E)=\pr(Y^{D}=1|X_{E}=1, X_{U\backslash E}=0)$ where $Y^{D}=\prod_{v\in D}Y_{v}$ is the binary random variable associated with the multimorbidity pattern $D$. More specifically, $\pr(Y^{D}=1)$ is the probability that the multimorbidity pattern $D$ occurs, and therefore (\ref{EQN:mob.regression}) is a sequence of regressions corresponding to the univariate binary responses $Y^{D}$ for $\emptyset\neq D\subseteq V$.

Similarly to Poisson regression, the parameters of LM regression
can be interpreted in terms of relative risks. We first consider the case $U=\{u\}$, so that $|U|=1$.
Hence, if we  denote the relative risk of an event $\mathcal{E}$ with respect to the two groups identified by $X_{u}=1$ and $X_{u}=0$ by
\begin{eqnarray}\label{EQN:relative.risk.factorization.002}
RR_{u}(\mathcal{E})=
\frac{\pr(\mathcal{E}|X_{u}=1)}
{\pr(\mathcal{E}|X_{u}=0)}
\end{eqnarray}
then for every $v\in V$ it holds that
$\beta^{\langle\mu\rangle}_{v}(u)=\log RR_{u}(Y_{v}=1)$ whereas, more generally, for $D\subseteq V$
\begin{eqnarray}\label{EQN.rel.risk.U1}
  \beta^{\langle\mu\rangle}_{D}(u)=\log RR_{u}(Y^{D}=1)=\log RR_{u}(\cap_{v\in D} \{Y_{v}=1\})
\end{eqnarray}
where we use the convention that $RR_{u}(Y^{\emptyset}=1)=1$.
\begin{exmp}[LM regression]\label{EXA:mob.reg.001}
Consider the case where $V=\{b, c, d\}$ and $U=\{h\}$. Then the saturated LM regression model contains a regression equation for every $D\subseteq V$, where the case $D=\emptyset$ is trivial. For $|D|=1$, equation (\ref{EQN:mob.regression}) has form
\begin{eqnarray*}
\log\mu_{i}(\emptyset)=\log \pr(Y_{i}=1|X_{h}=0)&=&\beta^{\langle\mu\rangle}_{i}(\emptyset)\\
\log\mu_{i}(h)=\log \pr(Y_{i}=1|X_{h}=1)&=&\beta^{\langle\mu\rangle}_{i}(\emptyset)+\beta^{\langle\mu\rangle}_{i}(h)
\end{eqnarray*}
where, for every $i\in V$, $\beta^{\langle\mu\rangle}_{i}(h)=\log RR_{h}(Y_i=1)$ that is the log-relative risk of the disease $Y_{i}$ for HIV-positive patients compared to HIV-negative patients. For $|D|=2$,
\begin{eqnarray*}
\log\mu_{ij}(\emptyset)=\log\pr(Y_{i}=1, Y_{j}=1|X_{h}=0)&=&\beta^{\langle\mu\rangle}_{ij}(\emptyset)\\
\log\mu_{ij}(h)=\log\pr(Y_{i}=1, Y_{j}=1|X_{h}=1)&=&\beta^{\langle\mu\rangle}_{ij}(\emptyset)+\beta^{\langle\mu\rangle}_{ij}(h)
\end{eqnarray*}
where, for every $i,j\in V$ with $i\neq j$,  $\beta^{\langle\mu\rangle}_{ij}(h)=\log RR_{h}(Y^{ij}=1)$ that is  the log-relative risk of the co-occurrence of the diseases $Y_{i}$ and $Y_{j}$ for HIV-positive patients compared to HIV-negative patients. Finally, for $|D|=3$ so that $D=V$,
\begin{eqnarray*}
\log\mu_{bcd}(\emptyset)=\log\pr(Y_{b}=1,Y_{c}=1, Y_{d}=1|X_{h}=0)&=&\beta^{\langle\mu\rangle}_{bcd}(\emptyset)\\
\log\mu_{bcd}(h)=\log\pr(Y_{b}=1,Y_{c}=1, Y_{d}=1|X_{h}=1)&=&\beta^{\langle\mu\rangle}_{bcd}(\emptyset)+\beta^{\langle\mu\rangle}_{bcd}(h)
\end{eqnarray*}
where $\beta^{\langle\mu\rangle}_{bcd}(h)=\log RR_{h}(Y^{bcd}=1)$ is the log-relative risk of the co-occurrence of the three diseases  for HIV-positive patients compared to HIV-negative patients.
\end{exmp}

Consider the regression equation relative to the subset $\emptyset\neq D\subseteq V$. It follows from (\ref{EQN:mob.regression}) that if for an element $u\in U$ it holds that $\beta^{\langle\mu\rangle}_{D}(E)=0$ for every $E\subseteq U$ such that $u\in E$, then $Y^{D}\ci X_{u}|X_{U\backslash \{u\}}$. This can be easily extended to a subset of covariates $X_{U^{\prime}}$ with $\emptyset\neq U^{\prime}\subseteq U$, whereas to generalize this result to the vector $Y_{D}$ it is necessary to consider a collection of regression equations as shown below.
\begin{prop}\label{THM:ind.equal.beta.zero}
Let $Y_{D}$ be the subvector of $Y_{V}$ indexed by $\emptyset\neq D\subseteq V$. For a subset $\emptyset\neq U^{\prime}\subseteq U$, it holds that $Y_{D}\ci X_{U^{\prime}}|X_{U\backslash U^{\prime}}$ if and only if $\beta^{\langle\mu\rangle}_{D^{\prime}}(E)=0$ for every $D^{\prime}\subseteq D$ and $E\subseteq U$ such that $E\cap U^{\prime}\neq\emptyset$.
\end{prop}
\begin{proof}
See Appendix A.
\end{proof}

We can conclude that it makes sense to focus on submodels characterized by zero LM regression coefficients because they encode interpretable relationships, possibly implying that one or more covariates have no-effect on the distribution of $Y^{D}$, or even on the joint distribution of $Y_{D}$. However, this approach did not identify any missing effects in the application to multimorbidity data. One reason for that is that the zero pattern of regression coefficients described in Proposition~\ref{THM:ind.equal.beta.zero} implies no-effect of $X_{U^{\prime}}$ on $Y_{D}$ and therefore that $X_{U^{\prime}}$ has no-effect on $Y_{v}$ for every single $v\in D$. Consider the case $|U|=1$ of a unique covariate $X_h$ representing the HIV-infection. It is well established that HIV has a relevant effect on each of the comorbidites considered singularly, that is, $\beta^{\langle\mu\rangle}_v(h)\neq 0$ for every $v \in V$. Although in principle it is possible for a covariate $X_{h}$ with a non-zero effect on $Y_{v}$ for some $v\in D$ to have coefficients equal to zero in the regression relative to $Y^{D}$, this hardly happens in practice because of the strong association existing between $Y^{D}$ and every $Y_{v}$ with $v\in D$; recall that $Y_{v}=0$ for any $v\in D$ implies $Y_{D}=0$ and, conversely, $Y^{D}=1$ implies $Y_{v}=1$ for every $v\in D$. In other words, it is well known that, if $X_{h}=1$ for HIV-postive patients, then $\pr(Y_{v}=1|X_{h}=1)>\pr(Y_{v}=1|X_{h}=0)$ for every comorbidity $v\in V$, and therefore it is reasonable to expect that also $\pr(Y^{D}=1|X_{h}=1)>\pr(Y^{D}=1|X_{h}=0)$ for every comorbidity pattern $D\subseteq V$, with $|D|>1$.

To disclose the usefulness of LM regression for the application considered in this paper, it is necessary to propose a different approach. In the multimorbidity analysis, it is useful to distinguish  between two different effects of HIV, specifically, one can be interest in the effect of  HIV  (i) on the prevalence of a comorbidity pattern $D$  and (ii) on the association among  comorbidities in $D$. As discussed above, the former is in general a well known matter  as multimorbidity shows a higher prevalence in infected patients. Indeed, the main question  is whether HIV plays a role in the way  single comorbidities combine together to produce the comorbidity pattern $D$, that is in the association of variables in $Y_D$.
We approach the problem by considering the extreme case of no-effect of HIV on the $D$-response associations because responses are conditionally independent given the covariates. Our standpoint is that if for a subset $D\subseteq V$ there exists a proper partition $A\cup B=D$ such that  $Y_{A}\ci Y_{B}|X_{U}$, then there is no-effect of the covariates on the $|D|$-way association of $Y_{D}$. This allows us to compare the performance of different link functions in multimorbidity analysis and, more importantly, to provide a clear interpretation to the effect of HIV in LM regression.
In the Example~\ref{EXA:mob.reg.002} below we illustrate how this idea can be formalized.
\begin{exmp}[LM regression vs. multivariate logistic regression]\label{EXA:mob.reg.002}
For the case $V=\{b, d\}$ and $U=\{h\}$, assume $Y_{b}\ci Y_{d}|X_{h}$ so that we can say that there is no-effect of $X_{h}=HIV$ on the association of $Y_{b}$ and $Y_{d}$. To see this in practice it is sufficient to consider any  measure of association that represents independence by a constant value, for instance the value zero, so that the association is the same for the two values of $X_{h}$. Exploiting the factorization of the probability function of  $Y_{bd}|X_{h}$ implied by conditional independence, it is not difficult to see that under the LM regression both
\begin{eqnarray*}\label{EQN:EXA:mob.reg.002.001}
Y_{b}\ci Y_{d}|X_{h}
\quad\Longleftrightarrow\quad
\beta^{\langle\mu\rangle}_{bd}(\emptyset)=\beta^{\langle\mu\rangle}_{b}(\emptyset)+\beta^{\langle\mu\rangle}_{d}(\emptyset)
\quad\mbox{and}\quad
\beta^{\langle\mu\rangle}_{bd}(u)=\beta^{\langle\mu\rangle}_{b}(h)+\beta^{\langle\mu\rangle}_{d}(h),
\end{eqnarray*}
and, under the multivariate logistic regression,
\begin{eqnarray*}
Y_{b}\ci Y_{d}|X_{h}\quad\Longleftrightarrow\quad\beta^{\langle\eta\rangle}_{bd}(\emptyset)=0
\quad\mbox{and}\quad
\beta^{\langle\eta\rangle}_{bd}(h)=0.
\end{eqnarray*}
The effect of HIV on the joint distribution of $Y_{bd}$ is represented by $\beta^{\langle\mu\rangle}_{bd}(h)$ in LM regression and by  $\beta^{\langle\eta\rangle}_{bd}(h)$ in multivariate logistic regression and therefore both
\begin{eqnarray}\label{EQN:betamusum}
    \beta^{\langle\mu\rangle}_{bd}(h)=\beta^{\langle\mu\rangle}_{b}(h)+\beta^{\langle\mu\rangle}_{d}(h)
\end{eqnarray}
and
\begin{eqnarray}\label{EQN:betaetazero}
  \beta^{\langle\eta\rangle}_{bd}(h)=0
\end{eqnarray}
can be used to state that there is no-effect of HIV on the association of $Y_{b}$ and $Y_{d}$. However, (\ref{EQN:betamusum}) and (\ref{EQN:betaetazero}) refer to different kinds of association because, although they are  necessary conditions for $Y_{b}\ci Y_{d}|X_{h}$, they are not sufficient for the same condition to hold true, and it can be easily checked that neither  (\ref{EQN:betamusum}) implies (\ref{EQN:betaetazero}) nor (\ref{EQN:betaetazero}) implies (\ref{EQN:betamusum}). From this perspective, it makes little sense to state that \emph{HIV has no-effect on the association of $Y_{b}$ and $Y_{d}$} if a clear interpretation to equalities (\ref{EQN:betamusum}) and (\ref{EQN:betaetazero}) is not provided.

The interpretation of equality (\ref{EQN:betamusum}) is straightforward. It states a connection between regressions of different order and implies that the effect of HIV on the distribution of $Y^{bd}$ is explained by the effects of HIV on the marginal distributions of $Y_{b}$ and $Y_{d}$. More interestingly, it provides an useful insight on the behaviour of relative risks because (\ref{EQN:betamusum}) is equivalent to
\begin{eqnarray}\label{EQN:exa.RR.factorization}
RR_{h}(\{Y_{b}=1\}\cap \{Y_{d}=1\})=RR_{h}(Y_{b}=1)\times RR_{h}(Y_{d}=1)
\end{eqnarray}
that is, if we are willing to interpret the effect of HIV by means of relative risks, then (\ref{EQN:betamusum}) allows one to carry out the analysis marginally on the regression equations for the main responses,
because the computation of the relative risk of the joint event $\{Y_{b}=1\}\cap\{Y_{d}=1\}$ does not require the joint distribution of $Y_{bd}|X_{h}$ but only the marginal distributions of $Y_{b}|X_{h}$ and $Y_{d}|X_{h}$ as in the case where $Y_{b}\ci Y_{d}|X_{h}$.

Multivariate logistic regression parameters are naturally associated with odds ratios, which play a very fundamental role among  association measures for categorical data. However, we deem that, in this context, the interpretation of (\ref{EQN:betamusum}) is more straightforward than  (\ref{EQN:betaetazero}). Indeed, the latter is equivalent to the identity $OR(Y_{bd}|\emptyset)=OR(Y_{bd}|h)$ that, unlike (\ref{EQN:exa.RR.factorization}), does not explains how the marginal effects of HIV on bone fracture and diabetes combine together to give the effect of HIV on the $\{b,d\}$-response association.
\end{exmp}

This example shows that the LM regression coefficients provide a clear way to investigate the effect of a covariate on the association of two binary variables and the following theorem generalises this result to $D$-response associations.
\begin{thm}\label{THM:independence.betamu}
Let $\mu$ be the mean parameter of $Y_{V}|X_{U}$ and $\beta^{\langle\mu\rangle}=\log(\mu)M_{U}$. Then, for a pair of disjoint nonempty subsets $A$ and $B$ of $V$ it holds that  $Y_{A}\ci Y_{B}|X_{U}$ if and only if  for every $D\subseteq A\cup B$ and $E\subseteq U$, with both $D\cap A\neq\emptyset$ and $D\cap B\neq\emptyset$, it holds that $\beta^{\langle\mu\rangle}_{D}(E)= \mathcal{B}^{\langle\mu\rangle}_{D}(E)$ where
\begin{eqnarray}\label{EQN:reference.beta}
  \mathcal{B}^{\langle\mu\rangle}_{D}(E)=-\sum_{D^{\prime}\subset D} (-1)^{|D\backslash D^{\prime}|}\; \beta^{\langle\mu\rangle}_{D^{\prime}}(E).
\end{eqnarray}
\end{thm}
\begin{proof}
See Appendix A.
\end{proof}
Theorem~\ref{THM:independence.betamu} shows that, whenever $Y_{D}$ can be split into two conditionally independent subvectors $Y_A$ and $Y_B$, then every coefficient of the LM regression with response $\log(\mu_{D})$ can be written as a linear combination of the corresponding coefficients in lower-order regressions. Hence, the difference between case and control patients with respect to $Y^{D}$ is not given by the effect of HIV on the association of all the comorbidities in $D$, but it is only the consequence of the effect that HIV has on the  occurrence of the subsets $A$ and $B$ of comorbidity patterns.

For the case $U=\{u\}$, the relationship $\beta^{\langle\mu\rangle}_{D}(u)=\mathcal{B}^{\langle\mu\rangle}_{D}(u)$ can be used to state  that $X_{u}$
has no-effect on the $D$-response association and, as well as in Example~\ref{EXA:mob.reg.002}, this statement means that the effect of $X_{u}$ on the distribution of $Y^{D}$ is explained by the corresponding coefficients in lower-order regressions. Furtheremore, it follows immediately from (\ref{EQN.rel.risk.U1}) that the equality $\beta^{\langle\mu\rangle}_{D}(u)=\mathcal{B}^{\langle\mu\rangle}_{D}(u)$ can be equivalently stated in terms of factorization of the relative risk $RR_{u}(Y^{D}=1)$ with respect to the collection of relative risks $RR_{u}(Y^{D^{\prime}}=1)$ for $D^{\prime}\subset D$.
\begin{exmp}[LM regression cont.]\label{EXA:mob.reg.003}
For the LM regression in Example~\ref{EXA:mob.reg.001} consider the case where $Y_{A}\ci Y_{B}|X_{h}$ where $A$ and $B$ are nonempty, disjoint subsets of $V$ such that $A\cup B=V$. Then, it follows from  Theorem~\ref{THM:independence.betamu} both that
\begin{eqnarray*}
\beta^{\langle\mu\rangle}_{bcd}(\emptyset)=-\beta^{\langle\mu\rangle}_{b}(\emptyset)-\beta^{\langle\mu\rangle}_{c}(\emptyset)-\beta^{\langle\mu\rangle}_{d}(\emptyset)+\beta^{\langle\mu\rangle}_{bc}(\emptyset)+\beta^{\langle\mu\rangle}_{bd}(\emptyset)+\beta^{\langle\mu\rangle}_{cd}(\emptyset)
\end{eqnarray*}
and that
\begin{eqnarray*}
&\beta^{\langle\mu\rangle}_{bcd}(h)=-\beta^{\langle\mu\rangle}_{b}(h)-\beta^{\langle\mu\rangle}_{c}(h)-\beta^{\langle\mu\rangle}_{d}(h)+\beta^{\langle\mu\rangle}_{bc}(h)+\beta^{\langle\mu\rangle}_{bd}(h)+\beta^{\langle\mu\rangle}_{cd}(h)&\\
\nonumber&\Updownarrow&\\
&RR_{h}(Y^{bcd}=1)=
\frac{RR_{h}(Y^{bc}=1)\times RR_{h}(Y^{bd}=1)\times RR_{h}(Y^{cd}=1)}
{RR_{h}(Y^{b}=1)\times RR_{h}(Y^{c}=1)\times RR_{h}(Y^{d}=1)}.&
\end{eqnarray*}
\end{exmp}

More generally, for the case $|U|>1$ it follows from Theorem~\ref{THM:independence.betamu} that the relationship
\begin{eqnarray}\label{EQN:no.effect.definition}
  \beta^{\langle\mu\rangle}_{D}(E)=\mathcal{B}^{\langle\mu\rangle}_{D}(E)\quad\mbox{for every }E\subseteq U\mbox{ such that } u\in E
\end{eqnarray}
implies that every regression coefficient involving $X_{u}$ is the linear combination of the corresponding coefficients of lower-order regressions, and can be used to state that $X_{u}$ has no-effect on the $D$-response association. Also in this case, (\ref{EQN:no.effect.definition}) can be interpreted in terms of relative risks. Consider the collection of relative risks of the event $Y^{D}=1$, with respect to $X_{u}$, conditionally on the values of the remaining covariates $X_{U\backslash \{u\}}$, given by
\begin{eqnarray}\label{EQN:relative.risk.E}
RR_{u}(Y^{D}=1|E)=
\frac{\pr(Y^{D}=1|X_{u}=1,\; X_{E}=1, X_{U\backslash (E\cup \{u\})}=0)}
{\pr(Y^{D}=1|X_{u}=0,\; X_{E}=1, X_{U\backslash (E\cup \{u\})}=0)}\quad\mbox{for }E\subseteq U\backslash\{u\};
\end{eqnarray}
we recall that we set $RR_{u}(Y^{\emptyset}=1|E)=1$. Then we  associate to every element of (\ref{EQN:relative.risk.E}) a reference value defined as follows.
\begin{defn}\label{DEF:reference.relative.risk} For $D\subseteq V$ with $|D|>1$, the \emph{reference relative risk} of the event $Y^{D}=1$ with respect to $X_{u}$ and a subset $E\subseteq U\backslash \{u\}$ is defined as
\begin{eqnarray*}
  \mathcal{RR}_{u}(Y^{D}=1|E)=\prod_{D^{\prime}\subset D}\; RR_{u}(Y^{D^{\prime}}=1| E)^{(-1)^{|D\backslash D^{\prime}|+1}}.
\end{eqnarray*}
\end{defn}
The following lemma states the connection between relative risks and regression coefficients as well as between reference relative risks in
Definition~\ref{DEF:reference.relative.risk} and Theorem~\ref{THM:independence.betamu}.
\begin{lem}\label{THM:lemma.beta.b}
Let $\mu$ be the mean parameter of $Y_{V}|X_{U}$ and $\beta^{\langle\mu\rangle}=\log(\mu)M_{U}$. Then, for every $D\subseteq V$, $u\in U$ and $E\subseteq U\backslash \{u\}$ it holds that
\begin{eqnarray*}
\log RR_{u}(Y^{D}=1|E)=\sum_{E^{\prime}\subseteq E} \beta^{\langle\mu\rangle}_{D}(E^{\prime}\cup \{u\})
\end{eqnarray*}
and, for $|D|>1$, that
\begin{eqnarray*}
\log\mathcal{RR}_{u}(Y^{D}=1|E)=\sum_{E^{\prime}\subseteq E} \mathcal{B}^{\langle\mu\rangle}_{D}(E^{\prime}\cup \{u\}).
\end{eqnarray*}
\end{lem}
\begin{proof}
See Appendix A.
\end{proof}

The introduction of the reference relative risks is motivated by the following result,
\begin{cor}\label{THM:independence.RR}
Let the subset $D\subseteq V$ be such that there exists a proper partition $D=A\cup B$, with $A, B \neq \emptyset$, satisfying $Y_{A}\ci Y_{B}|X_{U}$. Then for every $u\in U$ it holds that
\begin{eqnarray}\label{EQN:THM.independence.RR.001}
RR_{u}(Y^{D}=1|E)=\mathcal{RR}_{u}(Y^{D}=1|E)\quad\mbox{ for every }  E\subseteq U\backslash \{u\}.
\end{eqnarray}
\end{cor}
\begin{proof}
See Appendix A.
\end{proof}
Hence $\mathcal{RR}_{u}(Y^{D}=1|E)$ is a reference value in the sense that it is the value taken by the corresponding relative risk when $Y_{D}$ can be split into two conditionally independent subvectors. If (\ref{EQN:THM.independence.RR.001}) is satisfied, then all the conditional relative risks of the events $\{Y^{D}=1\}$, with respect to $X_{u}$, are equal to their reference value. This implies that, as far as the relative risk of the multimorbidity pattern $D$ is of concern, the analysis can be carried out marginally on the distributions of $Y_{D^{\prime}}|X_{U}$ with $D^{\prime}\subset D$. The following corollary shows that one can interpret (\ref{EQN:no.effect.definition}) in terms of relative risks because it is equivalent to (\ref{EQN:THM.independence.RR.001}).
\begin{cor}\label{THM:independence.RRfff}
Let $\mu$ be the mean parameter of $Y_{V}|X_{U}$ and $\beta^{\langle\mu\rangle}=\log(\mu)M_{U}$. For any $D\subseteq V$, with $|D|>1$, and $u\in U$ the relationship (\ref{EQN:no.effect.definition}) is satisfied if and only if (\ref{EQN:THM.independence.RR.001}) is satisfied.
\end{cor}
\begin{proof}
See Appendix A.
\end{proof}
We can conclude that LM regression provides a useful framework to investigate the effect of covariates on the association of responses. Submodels of interest involve possibly zero regression coefficients, as in Proposition~\ref{THM:ind.equal.beta.zero}, but, more commonly, regression coefficients which are linear combination of lower-order coefficients as in Theorem~\ref{THM:independence.betamu}. This approach can be easily implemented because (i) it involves submodels defined by linear constraints on the regression parameters and (ii) the computation of $\mathcal{B}^{\langle\mu\rangle}_{D}(E)$ in (\ref{EQN:reference.beta}) does not require to specify explicitly the partition of $Y_{D}$ into independent subvectors. A shortcoming of LM regression is that in order to interpret the values of regression coefficients in
terms of conditional independence the coefficients must be contrasted with the
theoretical values (\ref{EQN:reference.beta}) given in Theorem~\ref{THM:independence.betamu}. In the next section we show that LML regression provides a solution to this problem.

\subsection{Log-mean linear regression}
The LML regression model is obtained by setting $\theta$ in (\ref{EQN:betadef}) equal to the LML parameter $\gamma=M^{\T}_{V}\log(\mu)$ in (\ref{EQN:definition.of.gamma}), so that, in the saturated case, it follows from (\ref{EQN:beta-theta}) that
$\beta^{\langle\gamma\rangle}=\gamma M_{U}$ and $\gamma=\beta^{\langle\gamma\rangle}Z_{U}$; the latter can be written as
\begin{eqnarray}\label{EQN:LML.regression}
\gamma_{D}(E)=\sum_{E^{\prime}\subseteq E}\beta^{\langle\gamma\rangle}_{D}(E^{\prime})\quad\mbox{ for every }D\subseteq V, E\subseteq U.
\end{eqnarray}
There is a close connection between LM and LML regression given by a linear relationship between $\beta^{\langle\gamma\rangle}$ and  $\beta^{\langle\mu\rangle}$.
\begin{lem}\label{THM:betagamma.betamu}
Let $\mu$ and $\gamma$ be the mean  and LML parameter, respectively, of $Y_{V}|X_{U}$ so that  $\beta^{\langle\mu\rangle}=\log(\mu)M_{U}$ and $\beta^{\langle\gamma\rangle}=\gamma M_{U}$. Then it holds that $\beta^{\langle\gamma\rangle}=M_{V}^{\T}\beta^{\langle\mu\rangle}$, that is for every $D\subseteq V$ and $E\subseteq U$
\begin{eqnarray}\label{EQN.betagamma.betamu}
\beta^{\langle\gamma\rangle}_{D}(E)=\sum_{D^{\prime}\subseteq D} (-1)^{|D\backslash D^{\prime}|}\;\beta_{D^{\prime}}^{\langle\mu\rangle}(E).
\end{eqnarray}
\end{lem}
\begin{proof}
See Appendix A.
\end{proof}
As a consequence of Lemma~\ref{THM:betagamma.betamu}, submodels defined by zero LM regression coefficients as in Proposition~\ref{THM:ind.equal.beta.zero} can be equivalently stated by setting to zero the corresponding LML regression coefficients.
\begin{cor}\label{THM:cor.ind.equal.beta.zero}
Let $Y_{D}$ be the subvector of $Y_{V}$ indexed by $\emptyset\neq D\subseteq V$. For a subset $\emptyset\neq U^{\prime}\subseteq U$, it holds that $Y_{D}\ci X_{U^{\prime}}|X_{U\backslash U^{\prime}}$ if and only if $\beta^{\langle\gamma\rangle}_{D^{\prime}}(E)=0$ for every $D^{\prime}\subseteq D$ and $E\subseteq U$ such that $E\cap U^{\prime}\neq\emptyset$.
\end{cor}
\begin{proof}
See Appendix A.
\end{proof}

Furthermore, it follows immediately from (\ref{EQN.betagamma.betamu}) that for $|D|\leq 1$ it holds that $\beta^{\langle\gamma\rangle}_{D}=\beta^{\langle\mu\rangle}_{D}$ whereas for $|D|>1$ LML regression coefficients allow one to immediately check whether a LM regression coefficients coincides with the associated theoretical value given in Theorem~\ref{THM:independence.betamu}.
\begin{prop}\label{THM:gamma.equal.zero}
Let $\mu$ and $\gamma$ be the mean  and LML parameter, respectively, of $Y_{V}|X_{U}$ so that  $\beta^{\langle\mu\rangle}=\log(\mu)M_{U}$ and $\beta^{\langle\gamma\rangle}=\gamma M_{U}$. Then for every $D\subseteq V$, such that $|D|>1$, and $E\subseteq U$ it holds that $\beta^{\langle\gamma\rangle}_{D}(E)=\beta^{\langle\mu\rangle}_{D}(E)-\mathcal{B}^{\langle\mu\rangle}_{D}(E)$ so that
\begin{eqnarray*}
\beta^{\langle\gamma\rangle}_{D}(E)=0\quad\mbox{if and only if}\quad \beta^{\langle\mu\rangle}_{D}(E)=\mathcal{B}^{\langle\mu\rangle}_{D}(E).
\end{eqnarray*}
\end{prop}
\begin{proof}
This is an immediate consequence of (\ref{EQN.betagamma.betamu}).
\end{proof}
Hence, the relationship
\begin{eqnarray}\label{EQN:no.effect.gamma}
  \beta^{\langle\gamma\rangle}_{D}(E)=0\quad\mbox{for every }E\subseteq U\mbox{ such that } u\in E
\end{eqnarray}
is equivalent to both  (\ref{EQN:no.effect.definition}) and (\ref{EQN:THM.independence.RR.001}) and we can conclude that LML regression is equivalent to LM regression for the purposes of our analysis, with the advantage that all the submodels of interest can be specified by setting LML regression coefficients to zero. Furthermore, the value of the regression coefficients $\beta^{\langle\gamma\rangle}$ can be used to contrast relative risks with the corresponding reference values as follows.
\begin{cor}\label{THM:gamma.interpretation}
Let $\gamma$ be the  LML parameter of $Y_{V}|X_{U}$ so that $\beta^{\langle\gamma\rangle}=\gamma M_{U}$.  Then for every $D\subseteq V$, such that $|D|>1$, and $E\subseteq U$ it holds that
\begin{eqnarray}\label{EQN:ratio.of.RR}
\log\frac{RR_{u}(Y^{D}=1|E)}{\mathcal{RR}_{u}(Y^{D}=1|E)}=\sum_{E^{\prime}\subseteq E} \beta^{\langle\gamma\rangle}_{D}(E^{\prime}\cup \{u\})
\end{eqnarray}
\end{cor}
\begin{proof}
See Appendix A.
\end{proof}

Note that, interestingly, Corollary~\ref{THM:gamma.interpretation} implies that for the case $|U|=1$ the LML regression coefficients such that $|D|>1$ have an immediate interpretation as deviation of a relative risk from its reference value
\begin{eqnarray}\label{EQN:beta.gamma.single.covariate}
\beta^{\langle\gamma\rangle}_{D}(u)=\log \frac{RR_{u}(Y^{D}=1)}{\mathcal{RR}_{u}(Y^{D}=1)}.
\end{eqnarray}

We close this section by noticing that also the LM and the LML regression models belong to the family of generalized log-linear models. This allows us to exploit the asymptotic results of \citet{lang:1996} for the computation of maximum likelihood estimates (MLEs) and for model comparison. Nevertheless, we remark that the regression framework introduced in this paper is entirely novel.
Both the LM and the LML link functions allow us to specify the functional relationship between $\pi$ and the regression coefficients in closed form. This is a specific property that is not shared by the wider class of generalised log-linear models, and that makes it possible to specify closed-form functional relationships between regression coefficients of different responses. In turn, this is the basis for the factorization of relative risks with respect to the corresponding relative risks of lower order regressions.
\section{Analysis of multimorbidity data}\label{sec:app}
We now apply the LM and the LML regression models to the analysis of the multimorbidity data described in  Section~\ref{SEC:multimorbidity}. We consider four response variables:  $Y_{b}=\;$\emph{Bone fracture}, $Y_{c}=\;$\emph{Cardiovascular disease}, $Y_{d}=\;$\emph{Diabetes} and  $Y_{r}=\;$\emph{Renal failure}. Hence, $V=\{b,c,d,r\}$ and $Y_V$ takes values in $\{0,1\}^{4}$ where, for each variable, the value 1 encodes the presence of the disease. The four responses define 11 multimorbidity patterns denoted by the subsets $D\subseteq V$ with $|D|\geq 2$ and we say that $k=|D|$ is the \emph{size} of the multimorbidity pattern.

For the computation of MLEs we applied the algorithm and the maximization procedure given in \citet{lang:1996} properly adjusted
to account for the inclusion of the LM and LML link functions; for technical details and a review of further maximization approaches
see also \citet{eva-for:2013} and references therein.

It is well established that asymptotic methods are not efficient when the table of observed counts is sparse, that is when many cells have small frequencies \citep[see][Section~10.6]{agresti2013categorical}. For a given sample size, sparsity increases with the number of variables included in the analysis, therefore inference is  less reliable for response associatons of higher-order since they are computed on the relevant marginal tables. In order to keep sparsity at an acceptable level, we restrict the analysis to the effect of two binary covariates indexed by $U=\{a, h\}$; specifially, $X_{h}=\,$\emph{HIV} with the level 1 encoding the presence of the infection and $X_{a}=\,$\emph{Age} with the value 1 for patients aged 45 or more. Firstly, in Subsection \ref{subsec:appH}, we consider the regression model including only the covariate $X_h$, so that each regression coefficient represents the marginal effect of HIV on a $D$-response associations. Next, in Subsection \ref{subsec:appAH} we consider a regression model with two covariates including also the effect of age.

We remark that the case-control design for this study is not based on an outcome-dependent sampling. In fact the enrollment of each patient is independent of the disease status given the HIV status and further individual covariates. Therefore, for this case-study the population relative risk represents an identifiable measure of association, unlike case-control designs where the sample selection depends on the outcome of interest.
\subsection{The single covariate case}\label{subsec:appH}
When only the covariate $X_{h}$ is included in the model, the regression coefficients have a straightforward interpretation in terms of relative risks for cases versus controls. Indeed, in the  LM regression model, it holds that for $|D|=1$ the coefficient
$\beta^{\langle \mu \rangle}_D(h)$ is the log-relative risk of the occurrence of single commorbidities and, otherwise, it is the log-relative risk of the occurrence of the multimorbidity pattern $D$; see (\ref{EQN.rel.risk.U1}). In the LML regression, the coefficient $\beta^{\langle \gamma \rangle}_D(h)$ is equal to $\beta^{\langle \mu \rangle}_D(h)$ in case of single responses,
otherwise it is the log-ratio of the relative and the reference relative risks for the occurrence of the multimorbidity pattern $D$,  as shown in (\ref{EQN:beta.gamma.single.covariate}).  Hence, from Proposition \ref{THM:gamma.equal.zero}, the no-effect of HIV defined by  $\beta^{\langle \gamma \rangle}_D(h)=0$ with $|D|> 1$ implies that the relative risk equals the reference relative risk for the multimorbidity pattern $D$. Conversely, a positive or a negative value of $\beta^{\langle \gamma \rangle}_D(h)$ states that the relative risk for the pattern $D$ is higher or lower than its reference  relative risk, respectively.

As a  preliminary analysis, we provide in Table~\ref{tab.full_Hsat} the MLEs under the saturated  LM and LML regression models.
\setlength{\tabcolsep}{3pt}
\begin{table}[p]\small{
\centering
\caption{The saturated LML and LM regression models for  $Y_V|X_h$. The table gives the MLEs of the regression coefficients with their standard error (s.e.) and  $p$-value.}

\hspace*{-4mm}\begin{tabular}{l|rrrrrr|rrrrrr}
\hline
$D$         &  $\hat{\beta}^{\langle \gamma \rangle}_D(\emptyset)$  & s.e. & $p$-value & $\hat{\beta}^{\langle\gamma\rangle}_D(h)$ & s.e. & $p$-value & $\hat{\beta}^{\langle\mu\rangle}_D(\emptyset)$ & s.e. & $p$-value & $\hat{\beta}^{\langle\mu\rangle}_D(h)$ & s.e. & $p$-value  \\
\hline
$\{b\}$    &  -4.573 &   0.106 &   $<\!.001$    & 2.621 &  0.115  & $<\!.001$   & -4.573   & 0.106 & $<\!.001$  & 2.621  &0.115   &$<\!.001$  \\
$\{c\}$    & -4.476  &   0.101 &   $<\!.001$    & 1.056 &  0.143  & $<\!.001$   & -4.476   & 0.101 & $<\!.001$  & 1.056  & 0.143  &$<\!.001$    \\
$\{d\}$    & -3.255  &   0.054 &   $<\!.001$    & 1.061 &  0.075  & $<\!.001$   & -3.255   & 0.054 & $<\!.001$  & 1.061  & 0.075  &$<\!.001$   \\
$\{r\}$    & -6.321  &   0.251 &   $<\!.001$    & 3.570 &  0.261  & $<\!.001$   & -6.321   & 0.251 & $<\!.001$  & 3.570  & 0.261  &$<\!.001$  \\[1mm]
$\{b, c\}$ & 1.158   &   0.524 &   0.027    & -0.737&   0.558 & 0.187   & -7.892   & 0.541 & $<\!.001$  & 2.941  & 0.584  &$<\!.001$   \\
$\{b, d\}$ & 0.422   &   0.415 &   0.309    & -0.539&   0.437 & 0.217   & -7.407   & 0.430 & $<\!.001$  & 3.144  & 0.457  &$<\!.001$   \\
$\{b, r\}$ & 0.230   &   0.272 &   0.398    & -0.257&   0.325 & 0.429   & -10.665  & 0.005 & $<\!.001$  & 5.935  & 0.198  &$<\!.001$   \\
$\{c, d\}$ & 1.776   &   0.181 &   $<\!.001$    & -1.165&   0.268 & $<\!.001$   & -5.955   & 0.211 & $<\!.001$  & 0.953  & 0.309  &0.002   \\
$\{c, r\}$ & 0.133   &   0.270 &   0.622    & 0.404 &  0.393  & 0.304   & -10.665  & 0.005 & $<\!.001$  & 5.030  & 0.310  &$<\!.001$  \\
$\{d, r\}$ & 1.684   &   0.483 &   $<\!.001$    & -0.863&   0.497 & 0.082   & -7.892   & 0.541 & $<\!.001$  & 3.768  & 0.560  &$<\!.001$  \\[1mm]
$\{b, c, d\}$& -0.011  &   0.697 &   0.987    & -0.553&   0.909 & 0.543   & -8.960   & 0.909 & $<\!.001$  & 1.745  & 1.131  &0.123   \\
$\{b, c, r\}$ & 2.493   &  0.581  &  $<\!.001$    & -1.846&   0.676 & 0.006   & -11.358  & 0.005 & $<\!.001$  & 4.812  & 0.487  &$<\!.001$   \\
$\{b, d, r\}$ & 0.456   &  0.634  &  0.472    & -0.101&   0.681 & 0.882   & -11.358  & 0.005 & $<\!.001$  & 5.493  & 0.349  &$<\!.001$   \\
$\{c, d, r\}$ & -0.898  &  0.513  &  0.080    & 0.748 &  0.616  & 0.225   & -11.358  & 0.005 & $<\!.001$  &  4.812 &  0.487 &$<\!.001$   \\[1mm]
$\{b, c, d, r\}$ & -0.867  &  0.846  &  0.305    & 0.742 &  1.015  & 0.465   & -12.051  & 0.005 & $<\!.001$  & 4.143  & 0.952  &$<\!.001$   \\\hline
\end{tabular} \label{tab.full_Hsat}}
\end{table}
To clarify the meaning of the values given in this table consider, for instance, the disease pattern $\{b,c\}$, that is \emph{bone fracture--diabetes}. The estimated LM coefficient $\hat{\beta}^{\langle\mu\rangle}_{bc}(h)=2.941$ provides a strong evidence of the existence of an effect of HIV on the co-occurrence of the two diseases whereas the corresponding LML parameter estimate $\hat{\beta}^{\langle\gamma\rangle}_{bc}(h)=-0.737$ is not significantly different from zero, and this can be interpreted as no-effect of HIV on the association of the two diseases. More specifically, the estimated relative risk of the co-occurrence of bone fracture and diabetes takes value $\exp\{\hat{\beta}^{\langle\mu\rangle}_{bc}(h)\}=\exp(2.941)=18.9$ and is significantly different from 1; however, it is not significantly different from the product of the two marginal relative risks given by $\exp\{\beta^{\langle\mu\rangle}_{b}(h)\}\times\exp\{\beta^{\langle\mu\rangle}_{c}(h)\}$. The latter is the value taken by the relative risk of the co-occurrence of the two diseases when $Y_{b}\ci Y_{c}|X_{h}$ and therefore it implies that the relative risk of the co-occurrence of the two diseases only depends on the marginal relative risks of the two diseases rather than on the way the two diseases associate to give the multimorbidity pattern. Different is the case of the disease pattern $\{c, d\}$  because in this case the estimated LML coefficient $\hat{\beta}^{\langle\gamma\rangle}_{cd}(h)=-1.165$ is significantly different form zero and negative and therefore  the relative risk of the co-occurrence of the two diseases is significantly smaller than expected in the case the two diseases are conditionally independent given HIV. We remark that, when we say that a regression coefficient is significantly different from zero we refer to the statistical test, at level 5\%, based on the asymptotic normal  distribution of the MLEs.

As expected (see  the discussion following Proposition~\ref{THM:ind.equal.beta.zero})
in Table~\ref{tab.full_Hsat} the estimated log-relative risks $\hat{\beta}^{\langle \mu \rangle}_D(h)$ are positive for every single disease and for every multimorbidity pattern. Furthermore, all the corresponding regression coefficients are significantly different from zero, except for  the pattern $\{b,c,d\}$. The lowest relative risk corresponds to the pattern $\{c, d\}$, whereas the highest relative risks are in correspondence of the patterns $\{b, r\}$, $\{c, r\}$ and $\{b, d, r\}$.
More informative is the analysis of the LML regression coefficients. Indeed, most of the  coefficients $\beta^{\langle \gamma \rangle}_D(h)$ are not significantly different from zero, thereby providing an empirical evidence that the log-relative risks for the corresponding disease patterns are positive as a consequence of the effect of HIV on the single diseases they are composed, rather than for an effect of HIV on the associations of diseases. It is interesting that  most of the coefficients $\hat{\beta}^{\langle \gamma \rangle}_D(h)$, for $|D| \geq 2$, take a negative value. This is especially true for the multimorbidity patterns of size two, suggesting that the corresponding relative risks, although positive, take a value that is smaller than the value one would expect in case of conditional independence of the diseases.

The results provided by the saturated model can be summarized considering the average HIV-effect for disease patterns of the same size $k$, under the LM and LML approaches.
In particular, we compute the average of
$\hat{\beta}^{\langle \mu \rangle}_D(h)$ and of $\hat{\beta}^{\langle \gamma \rangle}_D(h)$, for $D$-response associations of the same size; see Appendix~B for further details.
Confidence intervals for these effects are plotted in Figure~\ref{fig:CI}.
\begin{figure}
\begin{center}
\includegraphics[scale=.7]{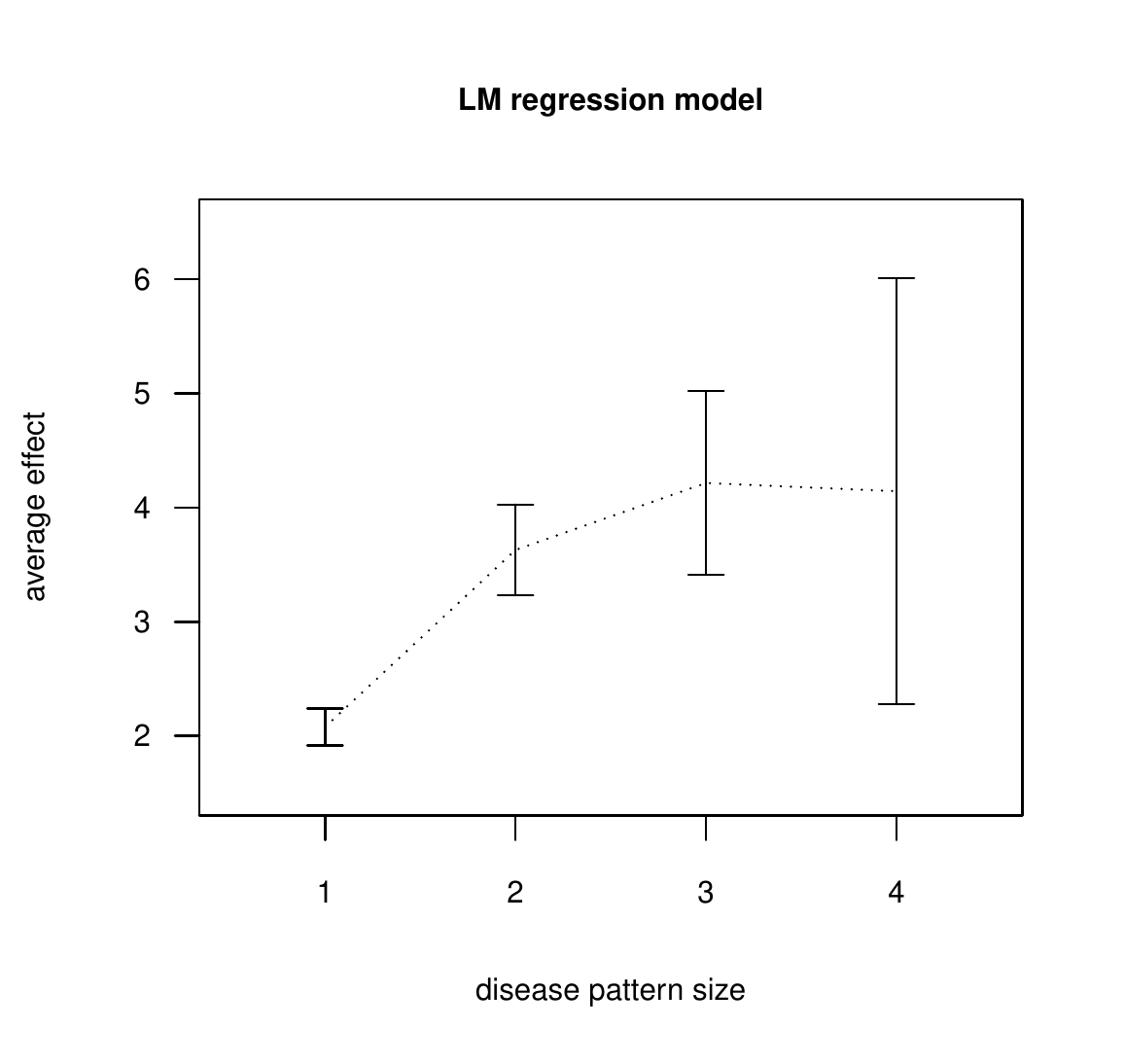}
\hspace*{5mm}
\includegraphics[scale=.7]{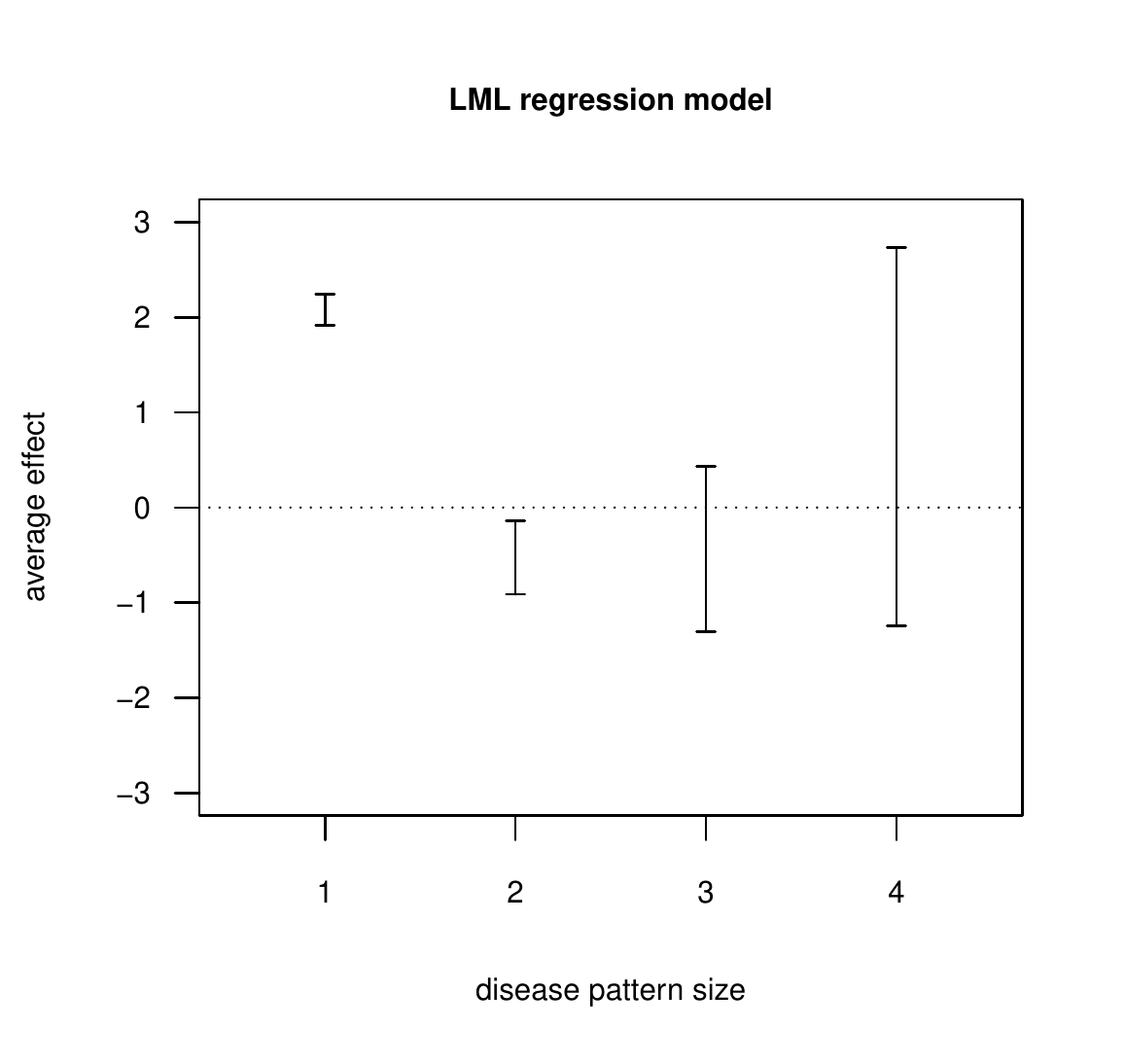}
\end{center}\caption{Confidence intervals of the average HIV-effect for response associations for disease patterns of the same size under the saturated LM and LML regression models.}\label{fig:CI}
\end{figure}
The plot in the left panel of Figure~\ref{fig:CI} gives the estimated LM average effect, i.e. the average log-relative risk for the occurrence of multimorbidity patterns of the same size, which  clearly increases with the size of the disease pattern.  On the other hand, the plot in the right panel of Figure~\ref{fig:CI} shows  the estimated LML average effect, that is the average effect of HIV on the association among diseases forming patterns of the same size. This effect appears to be negative for patterns of size $k=2$ and it might be null for $k=3$ and $k=4$.

Next, we apply the forward inclusion stepwise procedure described in Appendix B and select, in this way, the LML regression submodel given in Table~\ref{tab.full_H}. The latter  provides a very good fit of the data with a deviance 3.45 on 12 degrees of freedom and $p$-value, computed on the  asymptotic chi-square distribution of the deviance, equal to 0.99. In the selected model there is no effect of HIV on the associations $\{b,c\}$, $\{b,d\}$, $\{b,r\}$,  $\{c,r\}$, $\{d,r\}$, $\{b,c,d\}$ and $\{c,d,r\}$. Furthermore, it follows from Theorem~\ref{THM:independence.betamu} and Proposition~\ref{THM:gamma.equal.zero} that, under the selected model, three pairs of responses are conditionally independent given HIV, specifically, $Y_b \ci Y_d|X_h$, $Y_b \ci Y_r|X_h$ and $Y_c \ci Y_r|X_h$.

\begin{table}
\centering\caption{The selected LML, and the corresponding LM, regression model for  $Y_V|X_h$. The table gives the MLEs of the regression coefficients with their standard error (s.e.) and $p$-value.}
\small{
\begin{tabular}{l|rrrrrr|rrrr}
\hline
$D$        &  $\hat{\beta}^{\langle \gamma \rangle}_D(\emptyset)$  & s.e. & $p$-value & $\hat{\beta}^{\langle\gamma\rangle}_D(h)$ & s.e. & $p$-value & $\hat{\beta}^{\langle\mu\rangle}_D(\emptyset)$ & s.e. & $\hat{\beta}^{\langle\mu\rangle}_D(h)$ & s.e. \\
\hline
$\{b\}$    & -4.534  & 0.088    & $<\!.001$    & 2.579     &  0.100     & $<\!.001$   & -4.534  & 0.088  & 2.579& 0.100   \\
$\{c\}$    & -4.474  & 0.093    & $<\!.001$    & 1.056     & 0.138      & $<\!.001$   & -4.474  & 0.093  & 1.056& 0.138    \\
$\{d\}$    & -3.253  & 0.053    & $<\!.001$    & 1.058     & 0.075      & $<\!.001$   & -3.253  & 0.053  & 1.058& 0.075   \\
$\{r\}$    & -6.131  & 0.089    & $<\!.001$    & 3.383     & 0.114      & $<\!.001$   & -6.131  & 0.089  & 3.383& 0.114   \\[1mm]
$\{b, c\}$ & 0.443   & 0.181    & 0.014    & $\cdot$   & $\cdot$    & $\cdot$ & -8.565  & 0.227  & 3.635& 0.178   \\
$\{b, d\}$ & $\cdot$ & $\cdot$  & $\cdot$  & $\cdot$   & $\cdot$    & $\cdot$ & -7.787  & 0.101  & 3.637& 0.122   \\
$\{b, r\}$ & $\cdot$ & $\cdot$  & $\cdot$  & $\cdot$   & $\cdot$    & $\cdot$ & -10.665 & 0.011  & 5.962& 0.086   \\
$\{c, d\}$ & 1.777   & 0.124    & $<\!.001$    & -1.232    & 0.235      & $<\!.001$   & -5.950  & 0.182  & 0.882& 0.291   \\
$\{c, r\}$ & $\cdot$ & $\cdot$  & $\cdot$  & $\cdot$   & $\cdot$    & $\cdot$ & -10.605 & 0.118  & 4.439& 0.172   \\
$\{d, r\}$ & 0.840   & 0.102    & $<\!.001$    &$\cdot$    & $\cdot$    & $\cdot$ & -8.544  & 0.155  & 4.440& 0.142   \\[1mm]
$\{b, c, d\}$& $\cdot$ & $\cdot$  & $\cdot$  & $\cdot$   & $\cdot$    & $\cdot$ & -10.041 & 0.261  & 3.461& 0.299   \\
$\{b, c, r\}$ & 3.337  & 0.202    & $<\!.001$    & -2.631    & 0.431      & $<\!.001$   & -11.358 & 0.006  & 4.387& 0.439   \\
$\{b, d, r\}$ & 1.720  & 0.112    & $<\!.001$    & -1.382    & 0.238      & $<\!.001$   & -11.358 & 0.011  & 5.638& 0.251   \\
$\{c, d, r\}$ & $\cdot$& $\cdot$  & $\cdot$  & $\cdot$   & $\cdot$    & $\cdot$ & -11.241 & 0.222  & 4.264& 0.323   \\[1mm]
$\{b, c, d, r\}$& -1.777 & 0.124    & $<\!.001$    & 1.284     & 0.530      & 0.015   & -12.051 & 0.006  & 4.115& 0.729   \\\hline
\end{tabular} \label{tab.full_H}}
\end{table}
We now look more closely at the values taken by the estimated regression coefficients given in Table~\ref{tab.full_H}. We base this analysis on the asymptotic normal distribution of the MLEs,  but it is important to remark that, as we are dealing with post-selection parameter estimates, there might be distortions on the sampling distributions; see \citet{berk-al:2013} for a discussion and recent developments. The highest estimated relative risks are for multimorbidity patterns $\{b, r\}$ and $\{b, d, r\}$ with 95\% confidence intervals $(5.793; 6.130)$ and $(5.146; 6.130)$, respectively. More generally, high relative risks are given in correspondence of multimorbidity patterns including \emph{Renal failure}.

The negative values taken by the estimates $\hat{\beta}^{\langle \gamma \rangle}_D(h)$ for the remaining patterns of size two and three suggest that the relative risks of the multimorbidity patterns $\{c, d\}$, $\{b, d, r\}$ and $\{b, c, r\}$ are lower than their reference relative risks. For the disease pattern of size 4 the estimate of $\beta^{\langle \gamma \rangle}_D(h)$ is positive suggesting that the  relative risk of the pattern $\{b, c, d, r\}$ is  higher than its reference relative risk.
\subsection{The two-covariate case}\label{subsec:appAH}
We now introduce in the analysis of the previous section the additional covariate $X_{a}=Age$ and apply the model selection procedure
described in Appendix~B to obtain the model given in Table~\ref{tab.full_AH}. Such model has deviance $32.996$ on 33 degrees of freedom ($p=0.467$).
\begin{table}[h!]
\begin{center}\caption{The selected LML, and the corresponding LM, regression model for  $Y_V|X_{ah}$. The table gives the MLEs of the  regression coefficients with their standard error (s.e.) and $p$-value.}
\small{
\begin{tabular}{l|rrrrrrrrr|rr}
\hline
$D$         &  $\hat{\beta}^{\langle \gamma \rangle}_D(\emptyset)$  & s.e. & $p$-value & $\hat{\beta}^{\langle \gamma \rangle}_D(a)$ & s.e. & $p$-value& $\hat{\beta}^{\langle \gamma \rangle}_D(h)$ & s.e. & $p$-value & $\hat{\beta}^{\langle \mu \rangle}_D(h)$ & s.e. \\
\hline
$\{b\}$     & -4.717  & 0.117   & $<\!.001$    & 0.274  & 0.087  & 0.002  & 2.604  &  0.113 & $<\!.001$ &     2.604  &   0.113      \\
$\{c\}$     &  -5.347 &  0.170  &$<\!.001$     & 1.268  &  0.175 &$<\!.001$   &  1.044 & 0.140  & $<\!.001$ &     1.044  &   0.140    \\
$\{d\}$     &  -4.052 &  0.093  &$<\!.001$     & 1.159  &  0.095 & $<\!.001$  &  1.059 &0.074   &$<\!.001$  &     1.059  &   0.074\\
$\{r\}$   & -6.745  & 0.223   & $<\!.001$    & 0.869  & 0.153  & $<\!.001$  & 3.419  & 0.219  &$<\!.001$  &     3.419  &   0.219 \\[1mm]
$\{b, c\}$  & 1.292   & 0.273   & $<\!.001$    & $\cdot$ & $\cdot$& $\cdot$& -0.907 & 0.329  &0.006  &    2.742  &   0.363  \\
$\{b, d\}$  & $\cdot$ & $\cdot$ & $\cdot$  & $\cdot$ & $\cdot$& $\cdot$& $\cdot$& $\cdot$& $\cdot$&   3.663  &   0.135  \\
$\{b, r\}$  & 2.732   & 0.287   & $<\!.001$    & -0.720   &  0.322 & 0.025  & -2.223 &  0.361 & $<\!.001$ &   3.799  &   0.375   \\
$\{c, d\}$  &  2.421  & 0.308   &$<\!.001$     & -0.980   &  0.312 & 0.002  & -1.119 & 0.257  &$<\!.001$  &   0.985  &   0.297    \\
$\{c, r\}$  & 2.252   & 0.262   &$<\!.001$    & $\cdot$ & $\cdot$  & $\cdot$ & -1.881 &  0.379 & $<\!.001$ &  2.582  &   0.348   \\
$\{d, r\}$  & 2.249   & 0.294   &$<\!.001$     & -0.628   & 0.279  &0.024   & -1.053 &  0.320 &0.001  &   3.425  &   0.348   \\ [1mm]
$\{b, c, d\}$ & $\cdot$ & $\cdot$ & $\cdot$  & $\cdot$ & $\cdot$& $\cdot$& $\cdot$& $\cdot$& $\cdot$&   2.682  &   0.442\\
$\{b, c, r\}$ & $\cdot$ & $\cdot$ & $\cdot$  & $\cdot$ & $\cdot$& $\cdot$& $\cdot$& $\cdot$& $\cdot$ &  2.056  &   0.474   \\
$\{b, d, r\}$ & $\cdot$ & $\cdot$ & $\cdot$  & $\cdot$ & $\cdot$& $\cdot$& $\cdot$& $\cdot$& $\cdot$&   3.805  &   0.389     \\
$\{c, d, r\}$ & -1.312  & 0.376   & $<\!.001$    & $\cdot$ & $\cdot$& $\cdot$& 1.299  &0.485   &0.007  &    2.769  &   0.579     \\ [1mm]
$\{b, c, d, r\}$& $\cdot$ & $\cdot$ & $\cdot$  & $\cdot$ & $\cdot$& $\cdot$& $\cdot$& $\cdot$& $\cdot$&   2.242  &   0.661     \\\hline
\end{tabular} \label{tab.full_AH}}
\end{center}
\end{table}

When the model includes more than one covariate, the regression coefficients can be used to compute conditional relative risks as shown in Lemma~\ref{THM:lemma.beta.b} and Corollary~\ref{THM:gamma.interpretation}. However, in the model resulting from the application of the selection procedure  the interactions  $\beta_D^{\langle \gamma \rangle}(ah)$  are equal to zero for every $D\subseteq V$ and, by  (\ref{EQN.mon.Mobius}) and
(\ref{EQN.betagamma.betamu}), this implies that also $\beta_D^{\langle \mu \rangle}(ah)$ are equal to zero for every $D\subseteq V$. As a consequence, in the selected model the regression coefficients have a direct interpretation in term of conditional relative risks, for every $D\subseteq V$, as follows
\begin{eqnarray*}
\beta^{\langle\mu\rangle}_{D}(h)=\log RR_h(Y^D=1|\emptyset)=\log RR_h(Y^D=1|a)
\end{eqnarray*}
and similarly for $\beta^{\langle\mu\rangle}_{D}(a)$. LML regression coefficients coincide with LM regression coefficients for $|D|\leq 1$ whereas, for $|D|>1$, if follows from (\ref{THM:gamma.interpretation}) that
\begin{eqnarray*}
\beta^{\langle\gamma\rangle}_{D}(h)=\log \frac{RR_h(Y^D=1|\emptyset)}{\mathcal{RR}_h(Y^D=1|\emptyset)}=\log \frac{RR_h(Y^D=1|a)}{\mathcal{RR}_h(Y^D=1|a)}
\end{eqnarray*}
and similarly for $\beta^{\langle\gamma\rangle}_{D}(a)$.

Table~\ref{tab.full_AH} includes, for every $D \subseteq V$, the estimates $\hat{\beta}^{\langle \mu \rangle}_D(h)$ of the  corresponding LM regression model. As well as in the model with one covariate, also in this case  the fitted model provides high  relative risks of patterns involving \emph{Renal failure}. Nevertheless, compared with the results illustrated in Section \ref{subsec:appH}, the inclusion of \emph{Age} leads to a sensitive reduction of the estimated relative risk $\hat{\beta}^{\langle \mu \rangle}_D(h)$ for most of the multimorbidity patterns.

The fitted model shows a negative value of the estimates $\hat{\beta}^{\langle \gamma \rangle}_D(h)$ for most of the patterns of size two,  except for the pattern $D=\{b,d\}$ where the model constraints imply the conditional independence relationship $Y_b \ci Y_d|X_{ah}$. According to the selected model, HIV has no-effect on every $D$-response association of size $|D|>2$ with the exception of the pattern  $\{c,d,r\}$.

The estimates of the LML regression coefficients $\hat{\beta}^{\langle \gamma \rangle}_D(a)$
show a positive value in univariate regressions
as they coincide with the estimates of the log-relative risk of a single comorbidity for the increasing of age; in particular, the highest estimates  are for the events $\{Y_c=1\}$ and $\{Y_d=1\}$, whereas the lowest is for $\{Y_b=1\}$. For most of the multimorbidity patterns, the dataset supports the hypotheses of no-effect of age with $\hat{\beta}^{\langle \gamma \rangle}_D(a)=0$, except for patterns $\{b,r\}$, $\{c,d\}$ and $\{d,r\}$ where the estimates of the corresponding coefficients are negative.
\section{Discussion}\label{sec:discuss}
A wide range of alternative link functions for binary data have been proposed in the literature and, in particular, the marginal log-linear approach of \citet{ber-rud:2002} provides a wide and flexible class of links including the multivariate logistic one. However, in marginal regression modeling, iterative procedures need to be typically used to compute the cell probabilities from the parameters of the model. A key property of both the LM and the LML link functions is that probabilities can be analytically computed from the parameters. This inverse closed-form mapping is a distinguishing feature that confers our approach fundamental advantages that are not generally shared by other methods. In particular, this makes it possible to derive a closed-form functional relationships between the coefficients of regressions of different order that turns out to be very convenient when the interest is on the association of responses as in the multimorbidity application or, more generally, as stated by the line of enquiry (iii) given in \citet[][Sec.~6.5]{mcc-nel:1989}; see also Section~\ref{sec.intro}. More concretely, suitable collections of zero LML coefficients imply
the analysis in terms of relative risks can be equivalently performed in marginal distributions of the responses, even when independence relationships between responses do not hold true. Similar conclusions cannot be drawn by looking, for instance, at the coefficients of regressions based on the multivariate logistic link where analysis in terms of covariate effect on response associations can be equivalently carried out in marginal distributions only in case of independence.

With respect to the lines of enquiry given in Section~\ref{sec.intro}, we remark that
although the main focus of this paper is on the line of enquiry (iii),  the LM and the LML regression models are useful also when the interest is for the lines of enquiry (i) and (ii).
Firstly, with respect to (i), that is the analysis of the dependence structure of each response marginally on covariates, a central role is played by multivariate logistic regression because it maintains a marginal logistic regression interpretation for the single outcomes.
However, when the interest is for relative risks, rather than odds ratios,
the LM and LML regression models provide a useful alternative.
Secondly, when (ii) is of concern, that is a model for the joint distribution of all responses, Proposition~\ref{THM:ind.equal.beta.zero} and Corollary~\ref{THM:cor.ind.equal.beta.zero} show that the LM and LML regression models allow one to identify independencies among subsets of responses, conditionally on covariates. Interestingly, this feature is shared by the multivariate logistic regression \citep{mar-lup:2011} and, more precisely, the family of regression graph models \citep[see][]{wermuth2012sequencies} for binary data turns out as a special case of our approach. However, the regression graph representation of the model does not give the associations of responses and therefore we do not pursue this aspect here.

Future research directions for the class of LM and LML regression models involve the extension to response variables with an arbitrary number of levels, following a similar generalization of the LM and the LML parameterizations provided by \citet{roverato2015log}. It is also of interest the inclusion of continuous covariates and of additive random effects which may be
useful in the case where the sampling design induces unobserved correlation between
units which cannot be totally explained by the covariates. Unlike other approaches to regressions for categorical responses, such as the GEE models, the LM and LML regressions do not seem to lend themselves to semi-parametric fitting approaches because correlations between responses are regarded as parameters of interest rather than as nuisance parameters.

\section*{Acknowledgments}
We gratefully acknowledge Mauro Gasparini, Luca La Rocca and Nanny Wermuth for helpful discussions. We thank two Referees for their helpful comments.

\appendix

\section{Proofs}\label{sec.suppmat-proof}
\subsection*{Proof of Lemma~\ref{THM:lemma.beta.theta} }
(i)$\;\Rightarrow\;$(ii). Since $\theta_{D}(E)=\sum_{E^{\prime}\subseteq E} \beta^{\langle\theta\rangle}_{D}(E^{\prime})$, then if $E_{1},E_{2}\subseteq U$ it follows by (i) both that
\begin{eqnarray*}
\theta_{D}(E_{1})=\sum_{E^{\prime}\subseteq E_{1}} \beta^{\langle\theta\rangle}_{D}(E^{\prime})
         =\sum_{E^{\prime}\subseteq E_{1}\backslash U^{\prime}} \beta^{\langle\theta\rangle}_{D}(E^{\prime})
\end{eqnarray*}
and that
\begin{eqnarray*}
\theta_{D}(E_{2})=\sum_{E^{\prime}\subseteq E_{2}} \beta^{\langle\theta\rangle}_{D}(E^{\prime})
         =\sum_{E^{\prime}\subseteq E_{2}\backslash U^{\prime}} \beta^{\langle\theta\rangle}_{D}(E^{\prime}).
\end{eqnarray*}
Hence, $\theta_{D}(E_{1})=\theta_{D}(E_{2})$ whenever $E_{1}\backslash U^{\prime}=E_{2}\backslash U^{\prime}$, as required.

We now show that (ii)$\;\Rightarrow\;$(i). If $E\subseteq U$ is such that $E\cap U^{\prime}\neq\emptyset$ then we can find an element  $u\in E\cap U^{\prime}$ and, furthermore, it is straightforward to see that for every $E^{\prime}\subseteq E\backslash \{u\}$ it holds that $E^{\prime}\backslash U^{\prime}=(E^{\prime}\cup \{u\})\backslash U^{\prime}$ and this implies, by (ii), that  $\theta_{D}(E^{\prime})=\theta_{D}(E^{\prime}\cup \{u\})$. The result follows because
\begin{eqnarray*}
\beta^{\langle\theta\rangle}_{D}(E)
 &=&\sum_{E^{\prime}\subseteq E} (-1)^{|E\backslash E^{\prime}|} \theta_{D}(E^{\prime})\\
 &=&\sum_{E^{\prime}\subseteq E\backslash \{v\}} (-1)^{|E\backslash E^{\prime}|} \left\{\theta_{D}(E^{\prime})-\theta_{D}(E^{\prime}\cup \{v\})\right\}\\
 &=& 0.
\end{eqnarray*}

\subsection*{Proof of Proposition~\ref{THM:ind.equal.beta.zero} }
Recall that, since by construction $\beta^{\langle\mu\rangle}_{\emptyset}(E)=0$ for every $E\subseteq U$ it is sufficient to consider the case where $D^{\prime}\neq\emptyset$.
We first show that $\beta^{\langle\mu\rangle}_{D^{\prime}}(E)=0$ for every $\emptyset\neq D^{\prime}\subseteq D$ and $E\subseteq U$ such that $E\cap U^{\prime}\neq\emptyset$ implies that the conditional distribution
\begin{eqnarray}\label{EQN:proof.thm.ind.equal.beta.zero.002}
\pr(Y_{H}=1, Y_{D\backslash H}=0\mid X_{E}=1, X_{U\backslash E}=0)
\end{eqnarray}
does not depend on the value taken by $X_{U^{\prime}}$ for every $H\subseteq D$ and $E\subseteq U$. The result follows by noticing that for every $H\subseteq D$ and $E\subseteq U$ it holds that
\begin{eqnarray}\label{EQN:proof.thm.ind.equal.beta.zero.001}
  \pr(Y_{H}=1, Y_{D\backslash H}=0\mid X_{E}=1, X_{U\backslash E}=0)=
  \sum_{H^{\prime}\subseteq D\backslash H} (-1)^{|H^{\prime}|}\; \mu_{H\cup H^{\prime}}(E)
\end{eqnarray}
and that every term $\mu_{H\cup H^{\prime}}(E)$ in (\ref{EQN:proof.thm.ind.equal.beta.zero.001}) does not depend on the value taken by $X_{U^{\prime}}$. This follows from Lemma~\ref{THM:lemma.beta.theta} which states that $\beta^{\langle\mu\rangle}_{D^{\prime}}(E)=0$, for every $\emptyset\neq D^{\prime}\subseteq D$ and $E\subseteq U$ with $E\cap U^{\prime}\neq\emptyset$, implies
$\mu_{D^{\prime}}(E_{1})=\mu_{D^{\prime}}(E_{2})$, for every $E_{1},E_{2}\subseteq U$ with $E_{1}\backslash U^{\prime}=E_{2}\backslash U^{\prime}$. In turns this means that every term $\mu_{H\cup H^{\prime}}(E)$ in (\ref{EQN:proof.thm.ind.equal.beta.zero.001}) with $H\cup H^{\prime}\neq\emptyset$ does not depend on the value taken by $X_{U^{\prime}}$.  Furthermore, if $H\cup H^{\prime}=\emptyset$, then $\mu_{H\cup H^{\prime}}(E)=1$ by definition and this completes the first part of the proof.

We now prove the reverse implication, that is we assume that the conditional distribution of $Y^{D}|X_{U}$ in (\ref{EQN:proof.thm.ind.equal.beta.zero.002}) does not depend on $X_{U^{\prime}}$ and show that this implies
$\beta^{\langle\mu\rangle}_{D^{\prime}}(E)=0$ for every $\emptyset\neq D^{\prime}\subseteq D$ and  $E\subseteq U$ such that $E\cap U^{\prime}\neq\emptyset$. If $\emptyset\neq D^{\prime}\subseteq D$ then, by assumption, the conditional distribution of $Y^{D^{\prime}}|X_{U}$  does not depend on $X_{U^{\prime}}$ and, in turn, this  implies that $\mu_{D^{\prime}}(E_{1})=\mu_{D^{\prime}}(E_{2})$ for every $E_{1},E_{2}\subseteq U$ such that $E_{1}\backslash U^{\prime}=E_{2}\backslash U^{\prime}$ because, in this case, $\mu_{D^{\prime}}(E_{1})$ and $\mu_{D^{\prime}}(E_{2})$ are the probabilities of $Y_{D^{\prime}}=1$ conditioned on values of $X_{U}$ that may only differ for variables in $X_{U^{\prime}}$. Hence we can apply Lemma~\ref{THM:lemma.beta.theta} to obtain that $\beta^{\langle\mu\rangle}_{D^{\prime}}(E)=0$ for every  $E\subseteq U$ such that $E\cap U^{\prime}\neq\emptyset$.

\subsection*{Proof of Theorem~\ref{THM:independence.betamu}}
Theorem~1 of \citet{rov-lup-lar:2013} implies that $Y_{A}\ci Y_{B}|X_{U}$ if and only if every row of $M^{\T}\log(\mu)$ indexed by $D\subseteq A\cap B$ with $D\cap A\neq\emptyset$ and $D\cap B\neq\emptyset$ is equal to zero. Since $\log(\mu)=\beta^{\langle\mu\rangle}Z_{U}$ then the row of $M^{\T}\log(\mu)$ indexed by $D$ is equal to zero if and only if the corresponding row of
$M^{\T}\beta^{\langle\mu\rangle}Z_{U}$ is equal to zero. In turn, a row of $M^{\T}\beta^{\langle\mu\rangle}Z_{U}$ is equal to zero if and only if the corresponding row of $M^{\T}\beta^{\langle\mu\rangle}$ is equal to zero, because $Z_{U}$ has full rank. The row of $M^{\T}\beta^{\langle\mu\rangle}$ indexed by $D$ has entries
$\sum_{D^{\prime}\subseteq D} (-1)^{|D\backslash D^{\prime}|}\; \beta^{\langle\mu\rangle}_{D^{\prime}}(E)=\beta^{\langle\mu\rangle}_{D}(E)- \mathcal{B}^{\langle\mu\rangle}_{D}(E)$ for  $E\subseteq U$ and, therefore, it is equal to zero
if and only if   $\beta^{\langle\mu\rangle}_{D}(E)= \mathcal{B}^{\langle\mu\rangle}_{D}(E)$ for every $E\subseteq U$, as required.

\subsection*{Proof of Lemma~\ref{THM:lemma.beta.b} }
The first equality follows immediately from the fact that the relative risk is the ratio of two mean parameters as follows
\begin{eqnarray*}
\log RR_{u}(Y^{D}=1|E)=\log\left\{\frac{\mu_{D}(E\cup \{u\})}{\mu_{D}(E)}\right\}
=\sum_{E^{\prime}\subseteq E}\beta^{\langle\mu\rangle}_{D}(E^{\prime}\cup \{u\}).
\end{eqnarray*}
The second equality can be proved as follows,
\begin{eqnarray*}
\sum_{E^{\prime}\subseteq E}\mathcal{B}^{\langle\mu\rangle}_{D}(E^{\prime}\cup \{u\})
 &=&-\sum_{E^{\prime}\subseteq E}\sum_{D^{\prime}\subset D} (-1)^{|D\backslash D^{\prime}|}\; \beta^{\langle\mu\rangle}_{D^{\prime}}(E^{\prime}\cup \{u\})\\
 &=&-\sum_{D^{\prime}\subset D} (-1)^{|D\backslash D^{\prime}|}\; \sum_{E^{\prime}\subseteq E} \beta^{\langle\mu\rangle}_{D^{\prime}}(E^{\prime}\cup \{u\})\\
 &=&-\sum_{D^{\prime}\subset D} (-1)^{|D\backslash D^{\prime}|}\log RR_{u}(Y^{D^{\prime}}=1|E)\\
 &=&\log\prod_{D^{\prime}\subset D} RR_{u}(Y^{D^{\prime}}=1|E)^{(-1)^{|D\backslash D^{\prime}|+1}}\\
 &=&\log \mathcal{RR}_{u}(Y^{D}=1|E).
\end{eqnarray*}

\subsection*{Proof of Corollary~\ref{THM:independence.RR} }
By Lemma~\ref{THM:lemma.beta.b}
\begin{eqnarray*}
\log RR_{u}(Y^{D}=1|E)
=\sum_{E^{\prime}\subseteq E}\beta^{\langle\mu\rangle}_{D}(E^{\prime}\cup \{u\})
\end{eqnarray*}
and it follows from Theorem~\ref{THM:independence.betamu}, the fact that $Y_{A}\ci Y_{B}|X_{U}$, and  Lemma~\ref{THM:lemma.beta.b} that
\begin{eqnarray*}
\log RR_{u}(Y^{D}=1|E)
=\sum_{E^{\prime}\subseteq E}\mathcal{B}^{\langle\mu\rangle}_{D}(E^{\prime}\cup \{u\})
=\log \mathcal{RR}_{u}(Y^{D}=1|E).
\end{eqnarray*}

\subsection*{Proof of Corollary~\ref{THM:independence.RRfff} }
The fact that (\ref{EQN:no.effect.definition})$\Rightarrow$(\ref{EQN:THM.independence.RR.001})
follows immediately form Lemma~\ref{THM:lemma.beta.b}. To show the reverse implication, that is (\ref{EQN:THM.independence.RR.001})$\Rightarrow$(\ref{EQN:no.effect.definition})
we notice that in this case Lemma~\ref{THM:lemma.beta.b} implies
\begin{eqnarray}\label{EQN:corollary.last.001}
\sum_{E^{\prime}\subseteq E} \beta^{\langle\mu\rangle}_{D}(E^{\prime}\cup \{u\})=
\sum_{E^{\prime}\subseteq E} \mathcal{B}^{\langle\mu\rangle}_{D}(E^{\prime}\cup \{u\})
\mbox{ for every }E\subseteq U\backslash\{u\}
\end{eqnarray}
Hence, for $E=\emptyset$, (\ref{EQN:corollary.last.001}) implies
$\beta^{\langle\mu\rangle}_{D}(\{u\})=\mathcal{B}^{\langle\mu\rangle}_{D}(\{u\})$, and the proof follows by induction on the dimension of $E$.
We consider $E\subseteq U\backslash \{u\}$, with $E\neq\emptyset$,
and show that if the result is true for every $E^{\prime}\subset E$ then it also holds true for $E$. Indeed, in this case, it follows from (\ref{EQN:corollary.last.001}) that
\begin{eqnarray*}
\sum_{E^{\prime}\subset E} \mathcal{B}^{\langle\mu\rangle}_{D}(E^{\prime}\cup \{u\})+\beta^{\langle\mu\rangle}_{D}(E\cup \{u\})=
\sum_{E^{\prime}\subseteq E} \mathcal{B}^{\langle\mu\rangle}_{D}(E^{\prime}\cup \{u\})
\end{eqnarray*}
so that $\beta^{\langle\mu\rangle}_{D}(E\cup \{u\})=\mathcal{B}^{\langle\mu\rangle}_{D}(E\cup \{u\})$.

\subsection*{Proof of Lemma~\ref{THM:betagamma.betamu} }
By definition, $\gamma=M_{V}^{\T}\log(\mu)$ so that $\gamma M_{U}=M_{V}^{\T}\log(\mu)M_{U}$ and the result follows by recalling that $\beta^{\langle\gamma\rangle}=\gamma M_{U}$ and $\beta^{\langle\mu\rangle}=\log(\mu) M_{U}$.

\subsection*{Proof of Corollary~\ref{THM:cor.ind.equal.beta.zero} }
This is an immediate consequence of Proposition~\ref{THM:ind.equal.beta.zero}  and Lemma~\ref{THM:betagamma.betamu} because if, for $E\subseteq U$ it holds that $\beta^{\langle\mu\rangle}_{D^{\prime}}(E)=0$ for every $D^{\prime}\subseteq D$, then it follows from (\ref{EQN.betagamma.betamu}) that also $\beta^{\langle\gamma\rangle}_{D^{\prime}}(E)=0$ for every $D^{\prime}\subseteq D$. Conversely, it follows from Lemma~\ref{THM:betagamma.betamu} that $\beta^{\langle\mu\rangle}_{D}(E)=\sum_{D^{\prime}\subseteq D} \beta_{D^{\prime}}^{\langle\gamma\rangle}(E)$ and, therefore, if for $E\subseteq U$ it holds that $\beta^{\langle\gamma\rangle}_{D^{\prime}}(E)=0$ for every $D^{\prime}\subseteq D$ then also $\beta^{\langle\mu\rangle}_{D^{\prime}}(E)=0$ for every $D^{\prime}\subseteq D$.

\subsection*{Proof of Corollary~\ref{THM:gamma.interpretation} }
This follows from Proposition~\ref{THM:gamma.equal.zero} that
\begin{eqnarray*}
\sum_{E^{\prime}\subseteq E}\beta^{\langle\gamma\rangle}_{D}(E)=\sum_{E^{\prime}\subseteq E}\beta^{\langle\mu\rangle}_{D}(E)-\sum_{E^{\prime}\subseteq E}\mathcal{B}^{\langle\mu\rangle}_{D}(E)
\end{eqnarray*}
and the result follows from Lemma~\ref{THM:lemma.beta.b}.
\section{Technical details on the analysis of the multimorbidity data}\label{sec.suppmat-app}
\subsection*{The single covariate case}\label{subsec:suppmatH}
The confidence intervals in Figure~\ref{fig:CI} for the average effects of HIV on $D$-response associations of same size $k=|D|$  are obtained from the weighted averages of the MLEs of the saturated LM and LML regression models  given in Table~\ref{tab.full_Hsat}. For every $k=1,2,3,4$ we set
$$
\bar{\beta}^{\langle \mu \rangle}_k(h)=\sum_{D \subseteq V:|D|=k}w_{D}\hat{\beta}_{D}^{\langle \mu \rangle}(h) \quad \textrm{and} \quad \bar{\beta}^{\langle \gamma \rangle}_k=\sum_{D \subseteq V:|D|=k}w_{D}\hat{\beta}_{D}^{\langle \gamma \rangle}(h),
$$
where $w_D$, for $\emptyset\neq D\subseteq V$ is a suitable collection of weights. Concretely, given the observed sample, we set $w_D$ equal to the number of patients  affected by disease pattern $D$ normalized within the group of patients affected by disease patterns of
the same size. Hence, each $w_D$ represents the observed proportion of patients affected by the multimorbidity event $Y^D$ within the subgroup of patients affected by multimorbidity events of the same size $|D|$ so that, for every $k$,  $\sum_{D \subseteq V:|D|=k}w_{D}=1$.
It is worth noticing that $\bar{\beta}^{\langle \mu \rangle}_1(h)=\bar{\beta}^{\langle \gamma \rangle}_1(h)$ because for univariate regressions with $k=1$ the LM and the LML coefficients are equivalent,
and $\bar{\beta}^{\langle \mu \rangle}_4(h)=\hat{\beta}^{\langle \mu \rangle}_V(h)$ and $\bar{\beta}^{\langle \gamma \rangle}_4(h)=\hat{\beta}^{\langle \gamma \rangle}_V(h)$  as $V$ is the only pattern of size four.
Table~\ref{tab.index} collects these average effects with their standard errors and, to better understand the role played by these weights, one can check from Table~\ref{tab.full_Hsat} that the weighted means are closer to the medians of the averaged regression coefficients rather than to their arithmetic means.
\begin{table}[t]
\centering \caption{Average HIV-effect on $D$-response asssociations of the same size $k=1,\dots,4$, under the saturated LML and LM regression model for $Y_V|X_h$.
Standard errors are included in brackets.} 
\small{
\begin{tabular}{l|rrrr}
\hline
$k$&                                     \multicolumn1c{1}      &      \multicolumn1c{2}      &       \multicolumn1c{3}      &       \multicolumn1c{4}\\\hline
$\bar{\beta}^{\langle \gamma \rangle}_k(h)$ & $1.888$(0.064)  &   $-0.693$(0.203)  &   $-0.361$(0.460)   &   $0.742$(1.015) \\[1mm]
$\bar{\beta}^{\langle \mu \rangle}_k(h)$    & $1.888$(0.064)  &   $3.272$(0.225)   &   $4.614$(0.337)    &   $4.143$(0.952)\\
\hline
\end{tabular} \label{tab.index}}
\end{table}

The procedure adopted for model selection exploits the upward compatibility property satisfied
by the LML parameterization which implies that every term $\gamma_D(\cdot)$ of  $Y_V|X_h$
can be computed in the relevant marginal distribution of $Y_D|X_h$, with $D \subseteq V$.
Then, starting from univariate LML regression models, a forward procedure is followed which updates step-by-step the regression model for response associations of higher order such that  models for  associations of lower order (selected in the relevant marginal distributions) are preserved. In details, the  procedure  is based on four ordered steps: at every step $i$, with $i=1,\dots,4$, a LML regression models for $Y_D|X_h$ is selected for every $D \subseteq V$ with $|D|=i$. At every step, the structure, that is the zero terms, of the models chosen at the previous steps is maintained.
Reduced models at each step are selected setting zero constraints on regression coefficients not significatively different from zero.
Table~\ref{tab.univa_H}, \ref{tab.biva_H} and \ref{tab.triv_H} give  the MLEs of the  LML regression models selected at step 1, 2 and 3 of the
procedure, respectively. The estimates of the final  model selected at step 4 are shown in Table~\ref{tab.full_H}.

It is worth noticing that the estimates  $\hat{\beta}_D^{\langle\gamma\rangle}(E)$ are in general very close for the same $D$ and $E$ across Tables~\ref{tab.full_H}, \ref{tab.univa_H}, \ref{tab.biva_H} and \ref{tab.triv_H}. They present some differences because are MLEs computed under different models, but they are similar because they estimate the same quantity. Indeed, as a  consequence of the upward compatibility property and of the chosen selection procedure, the parameters $\beta_D^{\langle\gamma\rangle}(E)$ take the same values in the different models considered.
\subsection*{The two-covariate case}\label{subsec:suppmatAH}
The LML regression model for the distribution of $Y_V|X_{ah}$ is selected using a backward stepwise procedure which, starting from the saturated model, defines a sequence of nested models with zero constraints provided by the set of regression coefficients which are not significatively
different from zero. The procedure is based on three ordered steps in which three nested models are fitted, shortly denoted as $M_1$, $M_2$ and $M_3$.
\begin{itemize}
\item \emph{step 1}: model $M_1$ with no covariate interactions, i.e. with $\beta_D^{\langle\gamma\rangle}(ah)=0$ for every $D \subseteq V$, is fitted with a deviance 8.01 on 15 degrees of freedom and $p$-value equal to 0.923. Under $M_1$, some regression coefficients are not significatively different from zero such that the following constraints could be specified for models of $D$-response asssociations of size four,
    \begin{itemize}
    \item[-]  $\beta_D^{\langle\gamma\rangle}(a)=\beta_D^{\langle\gamma\rangle}(h)=0$, for $D=\{b,c,d,r\}$,
    \end{itemize}
    of size three,
    \begin{itemize}
    \item[-] $\beta_D^{\langle\gamma\rangle}(a)=0$, for $D=\{c,d,r\}$,
    \item[-] $\beta_D^{\langle\gamma\rangle}(\emptyset)=\beta_D^{\langle\gamma\rangle}(a)=\beta_D^{\langle\gamma\rangle}(h)=0$, for $D \in \{\{b,c,d\}, \{b,c,r\}, \{b,d,r\}\}$,
    \end{itemize}
    and of size two,
    \begin{itemize}
     \item[-] $\beta_D^{\langle\gamma\rangle}(a)=0$, for $D \in \{\{b,c\}, \{c,r\}, \{d,r\}\}$,
    \item[-] $\beta_D^{\langle\gamma\rangle}(\emptyset)=\beta_D^{\langle\gamma\rangle}(a)=\beta_D^{\langle\gamma\rangle}(h)=0$, for $D=\{b,d\}$.
    \end{itemize}
\item \emph{step 2}: model $M_2$ is fitted including  zero constraints for models of $D$-response associations of size four and three shown under $M_1$. The resulting model has a deviance 29.00 on 27 degrees of freedom and $p$-value equal to 0.413. Under  $M_2$, regression coefficients which are not significatively different from zero show that the following constraints could be further specified for models of $D$-response associations of size four,
    \begin{itemize}
    \item[-] $\beta_D^{\langle\gamma\rangle}(\emptyset)=0$, for $D=\{b,c,d,r\}$,
    \end{itemize}
    and of size two,
     \begin{itemize}
     \item[-] $\beta_{bc}^{\langle\gamma\rangle}(a)=0$, for $D \in \{\{b,c\}, \{c,r\}\}$,
    \item[-] $\beta_{bd}^{\langle\gamma\rangle}(\emptyset)=\beta_{bd}^{\langle\gamma\rangle}(a)=\beta_{bd}^{\langle\gamma\rangle}(h)=0$, for $D=\{b,d\}$.
    \end{itemize}
\item \emph{step 3}: model $M_3$ is fitted including all zero constraints shown under  $M_2$. The resulting model has a deviance 33.00 on 33 degrees of freedom and $p$-value equal to 0.467. Regression coefficients under $M_3$ are all significatively different from zero and they are collected in Table~\ref{tab.full_AH} of the paper.
\end{itemize}
\begin{table}
\centering \caption{LML regression models for $Y_D|X_h$, for every $D \subseteq V$ with $|D|=1$. The table gives the MLEs of the regression coefficients with their standard error (s.e.) and $p$-value.}\label{tab.univa_H}
\small{
\begin{tabular}{l|rrrrrr}
\hline
$D$         &  $\hat{\beta}^{\langle\gamma\rangle}_{D}(\emptyset)$  & s.e. & $p$-value & $\hat{\beta}^{\langle\gamma\rangle}_{D}(h)$ & s.e.     & $p$-value \\
\hline
$\{b\}$    & -4.577  & 0.106  & $<\!.001$  & 2.625   &  0.116   & $<\!.001$        \\
$\{c\}$    & -4.480  & 0.101  & $<\!.001$  & 1.056   & 0.143    & $<\!.001$    \\
$\{d\}$    & -3.256  & 0.054  & $<\!.001$  & 1.062   & 0.076    & $<\!.001$      \\
$\{r\}$    & -6.344  & 0.257  & $<\!.001$  & 3.591   & 0.267    & $<\!.001$  \\
\hline
\end{tabular}}
\end{table}
\begin{table}
\centering \caption{The selected LML regression models for $Y_D|X_h$, for every $D \subseteq V$ with $|D|=2$. The table gives the MLEs of the regression coefficients with their standard error (s.e.) and $p$-value, the deviance of the selected models with their degrees of freedom (d.f.) and  $p$-value.}\label{tab.biva_H}
\small{
\begin{tabular}{l|rrrrrrrrrr}
\hline
$D$         &  $\hat{\beta}^{\langle\gamma\rangle}_{D}(\emptyset)$  & s.e. & $p$-value & $\hat{\beta}^{\langle\gamma\rangle}_{D}(h)$ & s.e. & $p$-value &Deviance & d.f. & $p$-value\\
\hline
$\{b\}$    & -4.576  & 0.106    & $<\!.001$    & 2.623     &  0.115     & $<\!.001$        \\
$\{c\}$    & -4.479  & 0.101    & $<\!.001$    & 1.056     & 0.143      & $<\!.001$      \\
$\{b,c\}$ & 0.472   & 0.181    & 0.009    & $\cdot$   & $\cdot$    & $\cdot$   & 1.103  & 1 & 0.294\\\hline
$\{b\}$    & -4.576  & 0.106    & $<\!.001$    & 2.624     &  0.115     & $<\!.001$        \\
$\{d\}$    & -3.256  & 0.054    & $<\!.001$    & 1.061     & 0.076      & $<\!.001$      \\
$\{b,d\}$ & $\cdot$ & $\cdot$  & $\cdot$  & $\cdot$   & $\cdot$    & $\cdot$   & 1.489  & 2  & 0.475\\\hline
$\{b\}$    & -4.577  & 0.106    & $<\!.001$    & 2.625     &  0.116     & $<\!.001$        \\
$\{r\}$    & -6.344  & 0.257    & $<\!.001$    & 3.592     & 0.266      & $<\!.001$  \\
$\{b,r\}$ & $\cdot$ & $\cdot$  & $\cdot$  & $\cdot$   & $\cdot$    & $\cdot$   & 0.344  & 2 & 0.842\\\hline
$\{c\}$    & -4.479  & 0.101    & $<\!.001$    & 1.056     & 0.143      & $<\!.001$  \\
$\{d\}$    & -3.256  & 0.054    & $<\!.001$    & 1.062     &  0.076     & $<\!.001$        \\
$\{c,d\}$ & 1.773   & 0.182    & $<\!.001$    & -1.166    & 0.270      & $<\!.001$   & $\cdot$ & $\cdot$  \\\hline
$\{c\}$    & -4.480  & 0.101    & $<\!.001$    & 1.057     & 0.143      & $<\!.001$    \\
$\{r\}$    & -6.344  & 0.257    & $<\!.001$    & 3.592     & 0.266      & $<\!.001$  \\
$\{c,r\}$ & $\cdot$ & $\cdot$  & $\cdot$  & $\cdot$   & $\cdot$    & $\cdot$   & 2.946  & 2 & 0.229\\\hline
$\{d\}$    & -3.255  & 0.054    & $<\!.001$    & 1.060     &  0.076     & $<\!.001$        \\
$\{r\}$    & -6.336  & 0.256    & $<\!.001$     & 3.583     & 0.266      & $<\!.001$  \\
$\{d,r\}$ & 0.846   & 0.116    & $<\!.001$    &$\cdot$   & $\cdot$    & $\cdot$    & 1.764 & 1 & 0.184\\\hline
\end{tabular}}
\end{table}
\begin{table}
\centering \caption{The selected LML regression models for $Y_D|X_h$, for every $D \subseteq V$ with $|D|=3$. The table gives the MLEs of the regression coefficients with their standard error (s.e.) and $p$-value, the deviance of the selected models with their degrees of freedom (d.f.) and  $p$-value.}\label{tab.triv_H}
\small{
\begin{tabular}{l|rrrrrrrrrr}
\hline
$D$         &  $\hat{\beta}^{\langle\gamma\rangle}_{D}(\emptyset)$  & s.e. & $p$-value & $\hat{\beta}^{\langle\gamma\rangle}_{D}(h)$ & s.e. & $p$-value &Deviance & d.f & $p$-value  \\
\hline
$\{b\}$    & -4.574  & 0.106    & $<\!.001$    & 2.621     &  0.115     & $<\!.001$        \\
$\{c\}$    & -4.477  & 0.101    & $<\!.001$    & 1.056     & 0.143      & $<\!.001$      \\
$\{d\}$    & -3.255  & 0.054    & $<\!.001$    & 1.061     & 0.075      & $<\!.001$      \\
$\{b,c\}$ & 0.479   & 0.179    & 0.009    & $\cdot$   & $\cdot$    & $\cdot$      \\
$\{b,d\}$ & $\cdot$ & $\cdot$  & $\cdot$  & $\cdot$   & $\cdot$    & $\cdot$      \\
$\{c,d\}$ & 1.775   & 0.182    & $<\!.001$    & -1.166    & 0.268      & $<\!.001$       \\
$\{b,c,d\}$& $\cdot$ & $\cdot$  & $\cdot$  & $\cdot$   & $\cdot$    & $\cdot$   & 3.806 & 5   & 0.578 \\\hline
$\{b\}$    & -4.575  & 0.106    & $<\!.001$    & 2.620     &  0.115     & $<\!.001$        \\
$\{c\}$    & -4.478  & 0.101    & $<\!.001$    & 1.058     & 0.143      & $<\!.001$  \\
$\{r\}$    & -6.347  & 0.254    & $<\!.001$    & 3.596     & 0.263      & $<\!.001$  \\
$\{b,c\}$ & 0.439   & 0.183    & 0.016    & $\cdot$   & $\cdot$    & $\cdot$      \\
$\{b,r\}$ & $\cdot$ & $\cdot$  & $\cdot$  & $\cdot$   & $\cdot$    & $\cdot$   \\
$\{c,r\}$ & $\cdot$ & $\cdot$  & $\cdot$  & $\cdot$   & $\cdot$    & $\cdot$\\
$\{b,c,r\}$ & 2.910  & 0.345    & $<\!.001$    & -2.250    & 0.541      & $<\!.001$   & 4.643  & 5 & 0.461 \\\hline
$\{b\}$    & -4.576  & 0.106    & $<\!.001$    & 2.624     &  0.115     & $<\!.001$        \\
$\{d\}$    & -3.255  & 0.054    & $<\!.001$    & 1.061     & 0.075      & $<\!.001$      \\
$\{r\}$    & -6.344  & 0.253    & $<\!.001$    & 3.580     & 0.263      & $<\!.001$  \\
$\{b,d\}$ & $\cdot$ & $\cdot$  & $\cdot$  & $\cdot$   & $\cdot$    & $\cdot$      \\
$\{b,r\}$ & $\cdot$ & $\cdot$  & $\cdot$  & $\cdot$   & $\cdot$    & $\cdot$    \\
$\{d,r\}$ & 0.853   & 0.115    & $<\!.001$    &$\cdot$   & $\cdot$    & $\cdot$     \\
$\{b,d,r\}$ & 1.258 & 0.302     & $<\!.001$    & -0.940   & 0.369    & 0.011   & 3.505  & 5 & 0.623 \\\hline
$\{c\}$    & -4.479  & 0.101    & $<\!.001$    & 1.058     & 0.143      & $<\!.001$  \\
$\{d\}$    & -3.255  & 0.054    & $<\!.001$    & 1.058     &  0.076     & $<\!.001$        \\
$\{r\}$    & -6.336  & 0.254    & $<\!.001$    & 3.584     & 0.264      & $<\!.001$  \\
$\{c,d\}$ & 1.775   & 0.181    & $<\!.001$    & -1.233    & 0.272      & $<\!.001$   \\
$\{c,r\}$ & $\cdot$ & $\cdot$  & $\cdot$  & $\cdot$   & $\cdot$    & $\cdot$\\
$\{d,r\}$ & 0.831   & 0.117    & $<\!.001$    &$\cdot$   & $\cdot$    & $\cdot$\\
$\{c,d,r\}$ & $\cdot$ & $\cdot$  & $\cdot$  & $\cdot$   & $\cdot$    & $\cdot$   & 5.039  & 5 & 0.411 \\\hline
\end{tabular}}
\end{table}
\newpage
\bibliographystyle{chicago}
\bibliography{gamma-ref}

\end{document}